\documentclass[submission,copyright,creativecommons]{eptcs}
\providecommand{\event}{TLLA2020} 
\usepackage{breakurl}             
\usepackage{underscore}           

\title{Linear Additives}
\author{Gianluca Curzi
\institute{University of Turin\\
Torino, Italy}
\email{curzi@di.unito.it}
}
\def\titlerunning{Linear Additives}
\def\authorrunning{Gianluca Curzi 
}

\usepackage[T1]{fontenc} 
\usepackage[utf8]{inputenc} 
\usepackage[english]{babel} 
\usepackage{csquotes} 

\usepackage{amsfonts}
\usepackage{amssymb}
\usepackage{amsmath}
\usepackage{amsthm}

\usepackage{stmaryrd} 
\usepackage{mathrsfs} 
\usepackage{mathpartir} 
\usepackage{bussproofs} 
\usepackage{mathtools}
\usepackage{cmll} 

\usepackage[size=small,captionskip=20pt]{subfig} 
\usepackage[size=small,skip=15pt]{caption}
\usepackage{framed} 
\usepackage{graphicx} 
\usepackage{adjustbox} 
\usepackage{lscape}
\usepackage{afterpage}
\usepackage{rotating}

\usepackage{longtable} 
\usepackage[shortlabels]{enumitem} 

\usepackage{hyperref} 
\usepackage{xcolor} 
\definecolor{mygreen}{rgb}{0, 0.5, 0}
\definecolor{myred}{rgb}{0.5, 0, 0}
\definecolor{myblue}{rgb}{0, 0, 0.5}
\definecolor{lightblue}{rgb}{0.4, 0.6, 0.8}
\definecolor{myviolet}{rgb}{0.59, 0.29, 0}

\hypersetup{
    colorlinks,
    linkcolor={red!50!black},
    citecolor={blue!50!black},
    urlcolor={blue!80!black}  }
\usepackage{cancel}  

\usepackage[normalem]{ulem} 

\usepackage{tikz}
\usetikzlibrary{arrows.meta}
\usetikzlibrary{decorations.markings}
\usetikzlibrary{matrix}
\usetikzlibrary{arrows}
\usetikzlibrary{decorations.pathmorphing}
\usetikzlibrary{shapes}
\usetikzlibrary{arrows}
\usetikzlibrary{positioning} 
\usepackage{circuitikz} 
\usepackage{tikz-cd} 

\theoremstyle{plain}
\newtheorem{thm}{Theorem}
\newtheorem{lem}[thm]{Lemma}
\newtheorem{prop}[thm]{Proposition}
\newtheorem{cor}[thm]{Corollary}

\theoremstyle{definition}
\newtheorem{defn}{Definition}
\newtheorem{exmp}{Example}

\theoremstyle{remark}
\newtheorem{rem}{Remark}

\newcommand{\pcs}[1]{\langle \! \langle {#1}\rangle \! \rangle}

\newcommand{\s}[1]{\vert {#1}\vert}

\newcommand{\red}{\rightsquigarrow}

\begin{document}
\maketitle

\begin{abstract}
We introduce $\mathsf{LAM}$,  a subsystem of  $\mathsf{IMALL}_2$  with restricted  additive rules  able to manage duplication linearly, called  \textit{linear additive rules}. $\mathsf{LAM}$ is presented as the type assignment system for a  calculus endowed with copy constructors, which deal with substitution in a linear fashion. As opposed to the standard additive rules, the linear additive rules do not affect the complexity of term reduction: typable terms of $\mathsf{LAM}$  enjoy linear strong  normalization. Moreover, a mildly weakened version of cut-elimination for this system is proven which takes a cubic number of steps. Finally, we  define a  sound translation from  $\mathsf{LAM}$'s proofs into $\mathsf{IMLL}_2$'s linear lambda terms, and we study its complexity.

\end{abstract}

\section{Introduction}
\textit{Linear Logic} ($\mathsf{LL}$) is a refinement of both classical and intuitionistic  logic introduced by Girard in~\cite{girard1987linear}. A central role in $ \mathsf{LL} $ is played by the  \textit{exponential} modalities $\oc$ and $ \wn$, which give a logical status to the structural rules of weakening and contraction. The exponentials allow us to discriminate between \textit{linear} resources,  consumed exactly once, and \textit{non-linear} resources,  reusable at will. Moreover, since the uncontrolled use  of the structural rules is forbidden, conjunction and disjunction come with two distinct presentations:    the \textit{multiplicative} version (resp.~$\otimes$ and $ \parr$) and the \textit{additive} one (resp.~$\with$ and  $ \oplus $). 

The presence of multiple formulations of the same connective has prompted the analysis of specific  \emph{fragments} of Linear Logic, i.e.~subsystems of  $ \mathsf{LL} $  that illustrate the  behavior of a specific group of connectives. The simplest fragment of $ \mathsf{LL} $ is $\mathsf{MLL}$ (\textit{Multiplicative Linear Logic}),  i.e.~the modal-free subsystem of $ \mathsf{LL} $ with inference  rules for $ \otimes $, $ \parr $. Another one is  $\mathsf{MALL}$ (\textit{Multiplicative Additive Linear Logic}), obtained by extending $ \mathsf{MLL}$ with  \textit{additive rules}, i.e.~the inference rules for $ \with $, $ \oplus $.

As pointed out in~\cite{matsuoka2004nondeterministic}, according to the \enquote{computation-as-normalization} paradigm, the additive rules of $ \mathsf{LL} $ can be used to  express  non-deterministic program executions. This intuition has been further investigated  in the field of ICC (\textit{Implicit Computational Complexity}), a branch of computational complexity devising calculi that abstract from machine models and characterize complexity classes without imposing \enquote{external} resource bounds. In this setting, several variants of the additive rules able to express  non-deterministic computation more explicitly have been proposed to capture the class $\mathsf{NP}$. Examples are~\cite{matsuoka2004nondeterministic, gaboardi2008soft, maurel2003nondeterministic}, all based on \emph{light logics}, i.e.~subsystems of (second-order) $ \mathsf{LL} $ with weaker exponential rules that induce a complexity bound on proof normalization. 

Using variants of the additive rules to characterize the non-deterministic polytime problems raises some issues, because these inference rules affect the complexity of cut-elimination/normalization, which may require an exponential cost. A standard approach to circumvent this fact is to focus on a specific cut-elimination strategy called \textit{lazy}~\cite{girard2017proof}, which \enquote{freezes} those commuting cut-elimination steps that involve   duplication of (sub)proofs.    A similar technique can be found in~\cite{gaboardi2008soft}, where the  type system $\mathsf{STA}_+$ is introduced  to capture  the complexity class $\mathsf{NP}$ in the style of ICC. $\mathsf{STA}_+$ extends \textit{Soft Type Assignment}~\cite{gaboardi2009light}  with the \textit{sum rule} below 
\begin{equation}\label{eqn:  sum rule}
	\AxiomC{$\Gamma \vdash M:A$}
\AxiomC{$\Gamma \vdash N:A$}
\RightLabel{$sum$}
\BinaryInfC{$\Gamma \vdash M+ N:A$}
\DisplayProof
\end{equation}
   and the  non-deterministic reductions  $M \leftarrow M+N\rightarrow N$ for the choice operator   $+$. The sum rule is   close to the additive rules, and  suffers the same drawbacks. Consequently, to prove that  $\mathsf{STA}_+$ is sound for $\mathsf{NP}$, a particular reduction strategy is needed to avoid exponential computations. 

In this paper we present a different solution to the complexity-theoretical issues caused by the additive rules. We start focusing on the second-order intuitionistic formulation of $\mathsf{MALL}$, i.e.~$\mathsf{IMALL}_2$. We shall look at the latter  as a type system, essentially by considering formulas as types and by decorating logical derivations with $\lambda$-terms endowed with pairs and projections. The analysis of the non-linear features of the  additive rules   leads  to the new type system $\mathsf{LAM}$  (\textit{Linearly Additive Multiplicative Type Assignment}). This  is a subsystem of $\mathsf{IMALL}_2$  obtained by imposing some conditions on types and by replacing the standard additive rules with  weaker versions, called \textit{linear additive rules}. To some extent, $\mathsf{LAM}$ is more expressive than the \enquote{lazy} restriction of $\mathsf{IMALL}_2$. Indeed,  the linear additive rules  allow some forms of duplication that lazy evaluation  forbids. Nonetheless, these special additive rules are able to maintain control on the complexity of normalization, preventing exponential explosions and recovering \textit{linear} \textit{strong} normalization.  

The cut-elimination rules for $\mathsf{LAM}$ are constrained to copy-cat the reduction rules on terms, making the  cut rule  no longer  admissible. We then  identify  a class of types, called \textit{$\forall$-lazy},  whose derivations can always be turned into cut-free ones in cubic time.  This result is analogous to Girard's  restricted (lazy) cut-elimination theorem for $\mathsf{MALL}$ in~\cite{girard2017proof}, but somehow more permissive.

Last, following essentially~\cite{curzi2019type}, we introduce a computationally sound translation mapping a derivation of  $\mathsf{LAM}$  into a linear $\lambda$-term of $\mathsf{IMLL}_2$  whose  size can be exponentially bigger than the  source derivation. The translation exploits the mechanisms of linear duplication and erasure studied by Mairson and Terui in~\cite{mairsonlinear, mairson2003computational}, and shows how $\mathsf{LAM}$  is able to express such mechanisms in a very compact and elegant way.

\section{From Standard  Additives to Linear  Additives}
In this section we briefly recall the    $(\multimap, \with,  \forall)$ fragment of $\mathsf{IMALL}_2$ as a type  system,  and we show how nesting instances of the additive rules in a derivation  produces an exponential blow up in normalization. To circumvent this issue we introduce \emph{linear additives}, weaker versions of the standard additives permitting restricted forms of  duplication without affecting the complexity of normalization.

\subsection{The System  $\mathsf{IMALL}_2$} \label{subsec: IMALL}
We present   $\mathsf{IMALL}_2$  as a type assignment system for  the  calculus $\Lambda_{\pi, \langle \rangle}$, whose terms are defined by the following grammar:
 \begin{equation}\label{eqn: grammar of terms in IMALL}
M := x \ \vert \ \lambda x. M \ \vert \ MM \ \vert \ \langle M, M\rangle   \ \vert \ \pi_1 (M) \ \vert \ \pi_2(M)   
\end{equation}
where $x$ is taken from a denumerable set of variables. The set of free variables of a term $M$ is written $FV(M)$,   and  the  meta-level substitution for terms is denoted $M[N/x]$.  The   \textit{size} $\vert M \vert$ of a term $M$ is inductively defined as follows:
\begin{equation}\label{sizetermsMALL}
\begin{aligned}
\vert x \vert &\triangleq 1\\
\vert \lambda x. M \vert &\triangleq \vert \pi_i (M)\vert \triangleq \vert M \vert +1 && i \in \{ 1,2\} \\
\vert MN \vert &\triangleq \vert \langle M, N \rangle \vert \triangleq \vert M \vert + \vert N \vert +1.
\end{aligned}
\end{equation}
 The \emph{one-step relation $\rightarrow_{\beta }$} is the binary relation over $\Lambda_{\pi, \langle \rangle}$    defined by:
\begin{equation}\label{eqn: reduction rules lambda calculus with pairs}
\begin{aligned}
(\lambda x. M)N &\rightarrow_{\beta } M[N/x]  \qquad \qquad 
\pi_i\langle M_1, M_2\rangle & \rightarrow_{\beta } M_i \qquad  i \in \lbrace 1,2 \rbrace  
\end{aligned}
\end{equation}
Its reflexive and transitive closure is $\rightarrow^*_\beta$.  If  $M$ $\beta$-reduces to $N$  in exactly $n$ steps we write  $M \rightarrow^n_{\beta }N$. As usual, a $\lambda$-term is in (or is a)  \emph{normal form} whenever no reduction rule applies to it. 

The  set $\Theta_\with$ of types of $\mathsf{IMALL}_2$ is generated by the following grammar:
\begin{equation}
A:= \ \alpha \ \vert \ A \multimap A \ \vert \ A \with A \ \vert \ \forall \alpha. A    
\end{equation}
where $\alpha$ belongs to a denumerable set  of type variables. The set of  free type variables of a type $A$ is written $FV(A)$, and the standard meta-level  substitution for types is  denoted $A\langle B / \alpha \rangle$.  A type $A$ is \emph{closed} if $FV(A)=\emptyset$. The   \textit{size} $\vert A \vert$ of a type $A$  is inductively defined as follows:
\begin{equation}
	\begin{aligned}
	\vert \alpha \vert &\triangleq 1 \\
	\vert A \multimap B \vert &\triangleq  \vert A \with B\vert \triangleq  \vert A \vert + \vert B \vert +1 \\
	  \vert \forall \alpha . A\vert &\triangleq \vert A \vert + 1 .
	\end{aligned}
\end{equation}
We define the notions of \emph{positive subtype} and of \emph{negative subtype} of a type $A$ by simultaneous induction on the structure of $A$:
\begin{itemize}
\item $A$ is a positive subtype of itself;
\item if $B \multimap C$ is a positive (resp.~negative) subtype of $A$, then $B$ is a negative (resp.~positive) subtype of $A$, and $C$ is a positive (resp.~negative) subtype of $A$;
\item if $B \with C$ is a positive (resp.~negative) subtype of $A$, then so are $B$ and $C$;
\item if $\forall \alpha. B$ is a positive (resp.~negative) subtype of $A$, then so is $B$.
\end{itemize}
We say that a type $A$ has \emph{positive (resp.~negative) occurrences of $\forall$} if there exists a positive (resp.~negative) subtype of $A$ with shape $\forall \alpha. B$. We define in a similar way positive and negative occurrences of $\multimap$ and $\with$.

The system $\mathsf{IMALL}_2$ (\textit{Intuitionistic  Second-Order Multiplicative Additive Linear Logic}) is displayed in Figure~\ref{fig: The system IMALL2}, where   $\with$R and $\with$L$i$ are the  \textit{additives rules}. It derives \textit{judgements} with form $\Gamma \vdash M:A$, where  $M\in \Lambda_{\pi, \langle \rangle}$,  $A \in \Theta_\with$ and $\Gamma$ is a \textit{context}, i.e.~a finite multiset of \textit{assumptions} $x:A$.  The system requires the \textit{linear constraint}  $FV(\Gamma)\cap FV(\Delta)= \emptyset$ in both  $cut$  and $\multimap$L, where  $FV(\Gamma)$ denotes the set of all  free type variables in $\Gamma$.  With $\mathcal{D}\triangleleft \Gamma \vdash M: A$ we denote a derivation $\mathcal{D}$ of  $\Gamma \vdash M: A$, and in this case we say that $M$ is an \textit{inhabitant} of $A$ (or that $ A $ is \textit{inhabited} by $ M $). The \textit{size} $\vert \Gamma  \vert$ of a context $\Gamma= x_1: A_1, \ldots, x_n:A_n $ is $\sum^n_{i=1} \vert A_i \vert$, and the   \textit{size} $\vert \mathcal{D}\vert$ of a derivation $\mathcal{D}$ is the number of its rules applications.

We recall that  $\mathsf{IMLL}_2$ (\textit{Intuitionistic  Second-Order Multiplicative Linear Logic}) is obtained from  $\mathsf{IMALL}_2$ by excluding the additive rules from  Figure~\ref{fig: The system IMALL2}. It gives a type   \textit{exactly} to the class of \textit{linear $\lambda$-terms} (see~\cite{hindley1989bck,mairson2003computational}), i.e.~those terms $M$ from the standard $\lambda$-calculus such that: 
\begin{enumerate}[(i)]
\item each free variable of $M$ occurs in it exactly once, and
\item  for each subterm $\lambda x. N$ of $M$, $x$ occurs in $N$ exactly once. 
\end{enumerate}

 Tensors ($\otimes$) and  units  ($\mathbf{1}$) can be introduced in  $\mathsf{IMLL}_2$ (and hence in $\mathsf{IMALL}_2$)  by means of second-order definitions:
 \allowdisplaybreaks	
 \begin{equation}\label{eqn: tensor and unit}
 \begin{aligned}
&\textbf{1} \triangleq \forall \alpha. (\alpha \multimap \alpha) &  &A \otimes B \triangleq 
\forall \alpha. (A \multimap B \multimap \alpha) \multimap \alpha\\
&\mathbf I \triangleq \lambda x.x &  &M \otimes N \triangleq \lambda z. z\,M\,N \\
&\mathtt{let}\ M \mathtt{\ be \ }\mathbf I \mathtt{\ in \ }N \triangleq MN &&\mathtt{let}\ M \mathtt{\ be \ }x\otimes y \mathtt{\ in \ }N
\triangleq M(\lambda x. \lambda y. N) .
\end{aligned}	
 \end{equation}
 Hence, the  inference rules for $\otimes$ and $\mathbf{1}$, and the reduction rules 
 \allowdisplaybreaks	
 \begin{equation}\label{eqn: tensor unit reduction rules}
 \begin{aligned}
 \mathtt{let}\ \mathbf{I} \mathtt{\ be \ }\mathbf I \mathtt{\ in \ }N&\rightarrow_\beta N\\
 \mathtt{let}\  M_1\otimes  M_2  \mathtt{\ be \ }x_1\otimes x_2 \mathtt{\ in \ }N&\rightarrow_\beta N[M_1/x_1, M_2/x_2]
 \end{aligned}	
 \end{equation}
are  derivable in $\mathsf{IMLL}_2$.

\begin{figure}[t]
\centering
\begin{framed}
\begin{mathpar}
\inferrule*[Right=$ax$]{\\}{x: A \vdash x: A} \and \inferrule*[Right=$cut$]{\Gamma \vdash N: A \\ \Delta, x: A \vdash M: C}{\Gamma, \Delta \vdash M[N/x]: C}\\
\inferrule*[Right=$\multimap$R]{\Gamma, x: A \vdash M: B}{\Gamma \vdash \lambda x. M : A \multimap B}\and
\inferrule*[Right=$\multimap$L]{\Gamma \vdash N:A   \\ \Delta, x:B \vdash M:C}{\Gamma,  y: A \multimap B, \Delta \vdash M[yN/x]: C}  \\
 \inferrule*[Right=$\with$R]{\Gamma  \vdash M_1: A_1\\ \Gamma \vdash M_2: A_2}{\Gamma  \vdash \langle M_1, M_2\rangle : A_1 \with A_2}  \and 
 \inferrule*[Right=$\with$L$i$]{\Gamma, x_i: A_i \vdash M:C   \ \ \ \  i \in \lbrace1,2 \rbrace}{\Gamma,  y :A_1 \with A_2  \vdash M[\pi_i(y)/x_i] :C } \\
\inferrule*[Right=$\forall$R]{\Gamma \vdash M: A\langle \gamma /\alpha \rangle \ \ \ \ \gamma \not \in \mathrm{FV}(\Gamma)}{\Gamma \vdash M: \forall \alpha. A}\and
\inferrule*[Right=$\forall$L]{\Gamma, x: A \langle B/\alpha \rangle \vdash M: C}{\Gamma, x: \forall \alpha . A \vdash M: C}
\end{mathpar}
\caption{The system $\mathsf{IMALL}_2$.}
\label{fig: The system IMALL2}
\end{framed}
\end{figure}

\subsection{Exponential Blowup and Lazy Cut-Elimination for $ \mathsf{IMALL}_2$}\label{subsec:blowuplazy}

The additive rule $\with$R in Figure~\ref{fig: The system IMALL2} affects the complexity of normalization in $\mathsf{IMALL}_2$, letting the size of typable terms and the number of their  redexes grow exponentially during reduction. Definition~\ref{defn: nesting pairs} and Proposition~\ref{prop: lemma size/computation explosion} show an example. 


 \begin{defn}[Nesting $\mathsf{IMALL}_2$ terms]\label{defn: nesting pairs} Let  $A \in \Theta_\with$ and $M \in \Lambda_{\pi, \langle \rangle}$.  For all $n \in \mathbb{N}$, we define $A_{[n]}$ and $M_{[n]}$ as follows:
	\begin{equation}\label{eqn: type An and term Mn}
	A_{[n]} \triangleq \begin{cases}  A &n=0\\
	A_{[n-1]} \with A_{[n-1]} &n>0
	\end{cases}
	\qquad \qquad 
	M_{[n]} \triangleq \begin{cases}
	M &n=0\\
	\langle M_{[n-1]}, M_{[n-1]}\rangle &n>0
	\end{cases} 
	\end{equation}
	 Notice that $\langle M, M\rangle_{[n]}= M_{[n+1]}$. Moreover, for all $n \in \mathbb{N}$ we define $\mathtt{add}^x_{n}$  as follows:
	\begin{equation}\label{eqn: term add}
	\mathtt{add}^x_{n}\triangleq \begin{cases}x &n=0\\
	(\lambda y. \mathtt{add}^y_{n-1})\langle x,x\rangle &n>0
	\end{cases} 
	\end{equation}
\end{defn}

 \begin{prop}[Exponential blow up]\label{prop: lemma size/computation explosion}  Let $A\in \Theta_\with$,  and let $M\in \Lambda_{\pi, \langle \rangle}$ be of type $A$. For all $n \in \mathbb{N}$:
	\begin{enumerate}
		\item \label{enum: size explosion 1}   $\vdash  \lambda x.\mathtt{add}^x_{n}:A\multimap A_{[n]}$ is derivable in $\mathsf{IMALL}_2$;
		\item \label{enum: size explosion 2}  $(\lambda x.\mathtt{add}^x_n)M\rightarrow^{n+1}_\beta  M_{[n]}$, where $\vert (\lambda x.\mathtt{add}^x_n)M \vert= \mathcal{O}(n\cdot \vert M \vert) $ and $\vert M_{[n ]}\vert= \mathcal{O}(2^{n\cdot \vert M \vert})$.
	\end{enumerate}
\end{prop}
\begin{proof}
	Concerning point~\ref{enum: size explosion 1}, we prove by induction on $n $ that $\lambda x. \mathtt{add}^x_{n}$ has type  $A_{[k]} \multimap A_{[k + n]}$, for all $k \in \mathbb{N}$. If $n=0$ then  $\mathtt{add}^x_0=x$, so that $\lambda x.x$ has type $A_{[k]}\multimap A_{[k]}$. Let us consider the case  $n>0$.  By induction hypothesis,   $\lambda y. \mathtt{add}^y_{n-1}$ has type $A_{[k+1]}\multimap A_{[k+n]}$. If $x$ has type $A_{[k]}$ then  $\langle x, x\rangle $ has type $A_{[k+1]}$ and  $(\lambda y. \mathtt{add}^y_{n-1})\langle x,x\rangle $ has type $A_{[k+ n]}$. Therefore,   $\lambda x. \mathtt{add}^x_{n}=\lambda x. ((\lambda y. \mathtt{add}^y_{n-1})\langle x,x\rangle )$ has type $A_{[k]} \multimap A_{[k+n]}$.  Concerning point~\ref{enum: size explosion 2}, it suffices to prove by induction on $n $ that $\mathtt{add}^x_n\rightarrow^n_{\beta}x_{[n]}$. The base case is trivial. If $n>0$ then $\mathtt{add}^x_{n}= (\lambda y. \mathtt{add}^y_{n-1})\langle x,x\rangle\rightarrow_\beta \mathtt{add}^y_{n-1}[\langle x,x\rangle /y]$. Since by induction hypothesis $\mathtt{add}^y_{n-1}\rightarrow^{n-1}_\beta y_{[n-1]}$ then $\mathtt{add}^y_{n-1}[\langle x,x\rangle /x]\rightarrow^{n-1}_\beta y_{[n-1]}[\langle x,x\rangle /y]=\langle x,x\rangle_{[n-1]}=x_{[n]}$.  Finally,  we notice that, for all $n \in \mathbb{N}$, it holds that $\vert \mathtt{add}^x_n \vert= (5\cdot n)+1$ and $\vert M_{[n]} \vert = 2^{n \cdot \vert M \vert} + \sum^{n-1}_{i=0}2^{i}$. 
\end{proof}


Proposition~\ref{prop: lemma size/computation explosion} is about the complexity of normalization for  $\mathsf{IMALL}_2$'s typable terms,   but it can be  easily restated for cut-elimination. To prevent the   problem of exponential computation,  Girard introduced in~\cite{girard2017proof} the so-called \emph{lazy cut-elimination},  a rewriting procedure for both $\mathsf{MALL}$'s proof nets and their second-order versions that requires only a linear number of steps.  Intuitively, in the sequent calculus presentation of $\mathsf{IMALL}_2$,  lazy cut-elimination   \enquote{freezes} the commuting cut-elimination steps applied to those instances of $cut$ whose right premise is the conclusion of $\with$R, as they  involve duplication of (sub)proofs (see Figure~\ref{fig:commutativeMALL}). 
\begin{figure}[t]
	\centering
	\begin{framed}
		
		\begin{multline*}
				\AxiomC{$\mathcal{D}$}
			\noLine
			\UnaryInfC{$\Gamma \vdash M:A$}
			\AxiomC{\vdots}
			\noLine
			\UnaryInfC{$\Delta, x:A \vdash N: B$}
			\AxiomC{\vdots}
			\noLine
			\UnaryInfC{$\Delta, x:A \vdash P: C$}
			\RightLabel{$\with$R}
			\BinaryInfC{$\Delta, x:A \vdash \langle N, P \rangle: B\with C$}
			\RightLabel{$cut$}
			\BinaryInfC{$\Gamma, \Delta \vdash \langle N[M/x], P[M/x]\rangle : B \with C $}
			\DisplayProof\\
				\rightsquigarrow  \\
				\AxiomC{$\mathcal{D}$}
			\noLine
			\UnaryInfC{$\Gamma \vdash M:A$}
				\AxiomC{\vdots}
			\noLine
			\UnaryInfC{$\Delta, x:A \vdash N: B$}
			\RightLabel{$cut$}
			\BinaryInfC{$\Gamma, \Delta \vdash N[M/x]: B $}
				\AxiomC{$\mathcal{D}$}
			\noLine
			\UnaryInfC{$\Gamma \vdash M:A$}
			\AxiomC{\vdots}
			\noLine
			\UnaryInfC{$\Delta, x:A \vdash P: C$}
			\RightLabel{$cut$}
			\BinaryInfC{$\Gamma, \Delta \vdash P[M/x]: C $}
				\RightLabel{$\with$R}
				\BinaryInfC{$\Gamma, \Delta \vdash \langle N[M/x], P[M/x]\rangle : B \with C $}
				\DisplayProof
						\end{multline*}	

	\caption{Commuting cut-elimination step (with term annotation) involving duplication of $\mathcal{D}$.}
	\label{fig:commutativeMALL}
		\end{framed}
	\end{figure}
Since   lazy cut-elimination  cannot perform certain  cut-elimination steps, it may fail to produce a cut-free proof. Nonetheless, Girard showed that a cut-elimination theorem for $\mathsf{MALL}_2$'s proof nets holds if we stick to those receiving   ($\with, \exists$)-free types  (see~\cite{girard2017proof}). Following~\cite{mairson2003computational}, we call these special types  \enquote{lazy}, and we define their intuitionistic counterparts as follows:

\begin{defn}[Lazy types] A type $A \in \Theta_\with$ is \emph{lazy} if it contains no negative occurrences  of  $\forall$ and no positive occurrences of $\with$.
\end{defn}

Note that the restriction on negative occurrences of $\forall$ in the above definition is just to forbid the hiding of $\with$ connectives by $\forall$L in $\mathsf{IMALL}_2$. 

Lazy cut-elimination for $\mathsf{IMALL}_2$ can be reformulated as a reduction on terms by introducing  some   form of sharing, e.g.~explicit substitution  (see~\cite{DBLP:conf/popl/AbadiCCL90}). The idea is to replace the meta-notation $M[N/x]$ with a term constructor of the following shape:
\begin{equation}\label{eqn:explicitsubstitution}
M \pcs{ x: =N }
\end{equation}
and to add reduction rules   able to introduce and perform substitution  stepwise, such as:
\begin{equation}\label{eqn:suspendedsubstitution}
	\begin{aligned}
	(\lambda x. M)N &\rightarrow M  \pcs{ x: =N }\\
x  \pcs{ x: =N }&\rightarrow  N\\
y  \pcs{ x: =N }&\rightarrow  y &&y\neq x\\
(\lambda y.M) \pcs{ x: =N }&\rightarrow \lambda y.(M\pcs{ x: =N} ) &&y \neq x \text{ and }y \not \in FV(N)\\
(MP)\pcs{x:=N}&\rightarrow M\pcs{x:=N}P\pcs{x:=N}\\
\pi_i (M)\pcs{x:=N}&\rightarrow \pi_i(M\pcs{x:=N}) &&i\in \{1,2\}\\
\langle M, P \rangle  \pcs{ x: =N }&\rightarrow \langle M \pcs{ x: =N } , P \pcs{ x: =N }  \rangle .
	\end{aligned}
\end{equation}
\emph{Lazy reduction} is then obtained  by forbidding substitutions of  terms in a pair, i.e.~by ruling out the last rewriting rule of~\eqref{eqn:suspendedsubstitution}, since it mimics the cut-elimination step in Figure~\ref{fig:commutativeMALL}. As an example, consider the following typable terms of $\mathsf{IMALL}_2$:
\begin{equation}\label{eqn:MandM}
\begin{aligned}
M &\triangleq  (\lambda x. \pi_1\langle x, x \rangle )N\\
M'&\triangleq (\lambda x. \langle x, x \rangle )N.
\end{aligned}
\end{equation}
If we applied standard $\beta$-reduction to $M$ we would obtain $\pi_1\langle N, N \rangle$, thus creating  a (useless) copy of $N$. By contrast, thanks to explicit substitution, lazy reduction enables a better control on the parameter-passing mechanism and reduces $M$ without making new copies of $N$:
\begin{equation*}
(\lambda x. \pi_1\langle x, x \rangle )N\rightarrow \pi_1\langle x, x \rangle  \pcs{ x: =N } \rightarrow x \pcs{ x: =N } \rightarrow N.
\end{equation*}
The above reasoning does not apply to $M'$, as the  substitution $\langle x,x \rangle \pcs{ x: =N }$ would remain unperformed, so that lazy reduction may fail to  produce a substitution-free term, in general. Again, as in the case of cut-elimination, when a term has lazy type in $\mathsf{IMALL}_2$ each  pair  $\langle P,Q\rangle $  occurring in it  eventually turns into a redex $\pi_i\langle P, Q \rangle $,  so that all   suspended substitutions are sooner or later carried out.
Going back to the terms in~\eqref{eqn:MandM}, notice that  $M$  has lazy type in $\mathsf{IMALL}_2$  whenever  $N$ has, while only (non-lazy) types with shape $A \with A$ can be assigned to $M'$.

Explicit substitution and other forms of sharing  play a fundamental role in the study of reasonable cost models of the untyped  $\lambda$-calculus, where $\beta$-reduction and meta-level substitution cause size explosions similar to Proposition~\ref{prop: lemma size/computation explosion} (see e.g.~\cite{DBLP:journals/entcs/Accattoli18, DBLP:conf/rta/AccattoliL12, DBLP:conf/cie/LagoM06}). As an example, in~\cite{DBLP:conf/rta/AccattoliL12}  Accattoli and Dal Lago have shown  that $\lambda$-terms with explicit substitutions can be managed in reasonable (i.e.~polynomial) time, without having to unfold the sharing (that would re-introduce an exponential size blow up). Explicit substitutions have been introduced to cover the gap between the  $\lambda$-calculus and concrete implementations, but they can  also produce pathological behaviors in a typed setting, as shown by Melliès in~\cite{10.5555/645892.671591}.

\subsection{Linear Additives}\label{subsec: linear additives}
To prevent the exponential blow up discussed in the previous section we introduce weaker versions of  the standard additive rules called \emph{linear additive rules}, which give types to $\mathtt{ copy}$ constructs. The linear additive rules are  displayed  in  Figure~\ref{fig: linearadditives},  
\begin{figure}[t]
	\centering
	\begin{framed}
		\begin{mathpar}
				\inferrule*[Right=$\with$R$0$]{\vdash M_1: A_1 \\ \vdash M_2: A_2 }{\vdash \langle M_1, M_2 \rangle: A_1 \with A_2} \and
			\inferrule*[Right=$\with$L$i$]{\Gamma, x_i: A_i \vdash M:C   \ \ \ \  i \in \lbrace1,2 \rbrace}{\Gamma,  y :A_1 \with A_2  \vdash M[\pi_i(y)/x_i] :C } \\
				\inferrule*[Right=$\with$R$1$]{x_1: A \vdash M_1:A_1 \\ x_2: A \vdash M_2:A_2\\ \vdash V:A}{x: A \vdash \mathtt{copy}^{V}  x \mathtt{\ as \ }x_1, x_2 \mathtt{ \ in\ } \langle M_1,M_2 \rangle: A_1 \with A_2}  
		\end{mathpar}
		\caption{Linear additive rules, where $V\in \mathcal{V}$ and $A, A_1, A_2$ are closed and $ \forall$-lazy types.}
		\label{fig: linearadditives}
	\end{framed}
\end{figure}
where $V$ is a \emph{value}, i.e.~a closed and normal term free from instances of  $ \mathtt{copy} $ and projections, and the types $A, A_1, A_2$ are closed and \emph{$ \forall$-lazy}, according to the following definition:
\begin{defn}[$\forall$-lazy types] A type $A \in \Theta_\with$ is \emph{$\forall$-lazy} if it contains no negative occurrences of $\forall$. We say that  $x_1: A_1,\ldots , x_n: A_n \vdash M:B$ is a \textit{$\forall$-lazy judgement} if $A_1 \multimap \ldots \multimap A_n \multimap B$ is a $\forall$-lazy type.  Finally, we  say that $\mathcal{D}\triangleleft \Gamma \vdash M:A$ is a \textit{$\forall$-lazy derivation} if $ \Gamma \vdash M:A$ is a $\forall$-lazy judgement. 
\end{defn}
The reduction rule corresponding to $ \mathsf{copy} $ will be of the following form:
\begin{equation}\label{eqn: reduction rules LAML}
\mathtt{copy}^{V} U \mathtt{\ as \ }x_1, x_2 \mathtt{\ in \ }\langle M_1, M_2\rangle  \rightarrow \langle M_1[U/x_1], M_2[U/x_2]\rangle  \qquad U, V \in \mathcal{V} 
\end{equation}
where $\mathcal{V}$ is the set of all values. Notice that the inference rule $\with$R$0$ in Figure~\ref{fig: linearadditives} is introduced to let the above reduction rule  preserve types, i.e.~to assure Subject reduction. Intuitively, the operator $\mathtt{copy}$ behaves as a suspended substitution, quite like in  lazy reduction discussed in Section~\ref{subsec:blowuplazy}: the crucial  difference  is that lazy reduction  forbids  \emph{any} substitution of a term  $ N $ in a pair  $\langle M_1, M_2 \rangle$, while $\mathtt{copy}$ allows it when $N$ is turned into a value $U$.  

Hence, some limited forms of duplication are permitted by the linear additive rules. Nonetheless, they do not affect the complexity of normalization. On the one hand, indeed, the reduction rule in~\eqref{eqn: reduction rules LAML} can only copy values, i.e.~normal terms, so that no redex is duplicated. This allows  linear time normalization. On the other hand, since  any  \emph{$ \forall $-lazy} type $ A $ is inhabited by finitely many values (see Proposition~\ref{prop: properties of eta logn nf}.\ref{point: properties eta long nf 2}),  by  taking $V$ in Figure~\ref{fig: linearadditives} as the largest one of that type, we enable the size of   $\mathtt{copy}^V$ to bound the size of the new copy of $U$ produced by  applying~\eqref{eqn: reduction rules LAML} (since $U$ has type $A$). This lets  reduction strictly decrease the size of terms, recovering linear space normalization.

As a final remark, let us observe that the linear additive rule $\with$R$1$ introduces a seemingly severe restriction:  context-sharing is permitted for premises having \emph{exactly} one assumption.  This constraint has no real impact on the  algorithmic expressiveness of linear additives, since a general inference rule with premises sharing  multiple  assumptions can be easily derived by exploiting the definitions in~\eqref{eqn: tensor and unit}. Indeed, tensors are able to turn a context with $n$ assumptions $x_1: A_1, \ldots, x_n: A_n$ into one with single assumption $x: A_1\otimes \ldots \otimes A_n$.  Narrowing context-sharing has its benefits: it avoids heavy notation produced by several occurrences of the construct $\mathtt{copy}$, each one expressing the sharing of a single assumption.

\section{A Type Assignment With Linear Additives} 
In this section we introduce \textit{Linearly Additive Multiplicative Type Assignment}, $\mathsf{LAM}$ for short. It  is a subsystem of $\mathsf{IMALL}_2$  endowed with the linear additive rules in Figure~\ref{fig: linearadditives}. Thanks to the controlled use of substitution, these rules can be freely nested in $\mathsf{LAM}$ without incurring exponential normalization. 

\subsection{The System $\mathsf{LAM}$}\label{subsec: LAM}
The following grammar generates \textit{raw terms}:
 \begin{align}
 M, N &:= \, x \ \vert \ \lambda x. M \ \vert \ MM  \ \vert \ \langle M, M \rangle \ \vert \ \pi_i(M) \ \vert \ \mathtt{copy}^{U} M \mathtt{\ as \ }x, y \mathtt{\ in \ } \langle M, M\rangle \label{eqn:rawterms}\\
 U, V&:= \, x \ \vert \ \lambda x. U \ \vert \ UU \ \vert \ \langle U, U \rangle  \label{eqn:values}
 \end{align}
where $\mathtt{copy}^{U} M \mathtt{\ as \ }x, y \mathtt{\ in \ } \langle P, Q\rangle$ binds both $x$ in $P$ and $y$ in $Q$. A \emph{value} is any closed raw term generated by the grammar~\eqref{eqn:values} that is \emph{normal} with respect to the reduction step $(\lambda x. U)V \rightarrow U[V/x]$. The set of all values is denoted $\mathcal{V}$. A raw term is a \emph{term}  if any occurrence of the $ \mathtt{copy}^{U}  $ operator in it is such that $U \in \mathcal{V}$. The set of terms is denoted $\Lambda_{\mathtt{copy}}$. We extend the definition of  $\vert M \vert $  in~\eqref{sizetermsMALL}  to the new clause:
\allowdisplaybreaks
\begin{align*}
\s{\mathtt{copy}^{U} M \mathtt{\ as \ }x, y \mathtt{ \ in \ }\langle P, Q \rangle } &= \vert U \vert + \vert M \vert + \vert  \langle P,    Q \rangle  \vert + 1 
\end{align*}
 The \textit{one-step reduction}  $\rightarrow$ on $\Lambda_{\mathtt{ copy}}$  extends $\rightarrow_{\beta}$ in~\eqref{eqn: reduction rules lambda calculus with pairs}  with the reduction rule in~\eqref{eqn: reduction rules LAML}, 
 and applies in any context. Its reflexive and transitive closure is  $\rightarrow^*$. If $M$ reduces to $N$ in exactly $n$ steps we write $M \rightarrow^n N$.  A term is in (or is a)  \textit{normal form} if no reduction applies to it.

The system $\mathsf{LAM}$ is essentially obtained by replacing the standard additive rules of $\mathsf{IMALL}_2$  with the linear additive rules in Figure~\ref{fig: linearadditives}, and by imposing a mild requirement on the inference rules $\multimap$L and $\forall$R that plays an crucial role to assure a weak form of cut-elimination.
\begin{defn}[The system $\mathsf{LAM}$]  \label{defn: the system LAM}  \label{defn: eta long nf} 
 The system $\mathsf{LAM}$  (\textit{Linearly Additive Multiplicative Type Assignment}) is the type assignment for $\Lambda_{\mathtt{copy}}$ obtained by extending $\mathsf{IMLL}_2$ with  the \textit{linear additive rules} in Figure~\ref{fig: linearadditives}, and by imposing the following \emph{closure conditions}:
\begin{enumerate}[(i)]
\item no instance of  $\multimap$L  has conclusion $\Delta, y: A \multimap B, \Gamma \vdash M:C$ with $FV(B)= \emptyset$ and $FV(A)\neq \emptyset$;
\item    no instance  of $\forall$R has conclusion $\Gamma \vdash M :\forall \alpha. A$ with  $FV(\forall \alpha. A)= \emptyset$ and $FV(\Gamma)\neq  \emptyset$.
\end{enumerate} 
 \end{defn}

As it stands, $\mathsf{LAM}$ is able to linearly bound the number of steps required to normalize typable terms, essentially because we can only duplicate values, which are redex-free. However, the system has no control over the size of terms, which may grow exponentially during normalization. What we need is to bound the size of the new copies of values that are produced by applying the reduction rule~\eqref{eqn: reduction rules LAML}. 
This is the goal of Proposition~\ref{prop: properties of eta logn nf} and Remark~\ref{rem: eta long nf}.

\begin{defn}[$\eta$-expansion] \label{defn: eta long nf} Given $\mathcal{D}\triangleleft \Gamma \vdash M:A$ a $\forall$-lazy and \emph{cut-free} derivation, we say that $\mathcal{D}$ is   \textit{$\eta$-expanded} if all its  axioms  are \textit{atomic}, i.e.~of the form $x:\alpha \vdash x: \alpha$ for some type variable $\alpha$. In this case, $M$ is called a \textit{$\eta$-long normal form} (note that $M$ must be a normal form).
\end{defn}
We now state some basic properties about $\forall$-lazy  and $\eta$-expanded derivations.

\begin{prop}[$\forall$-lazy derivations]\label{prop: lazyness propagation}{\ }
	\begin{enumerate}
		\item \label{point 1: laziness propagation} If a premise of one  among the rules $\multimap$R, $\multimap$L and  $\forall$R   is not $\forall$-lazy, then neither is its conclusion. Moreover, the conclusions of $\with$R$1$, $\with$L$i$ and  $\forall$L  are not $\forall$-lazy, while  the conclusion of $\with$R$0$ is.
		\item  \label{point 2: laziness propagation}  Let     $\mathcal{D}\triangleleft \Gamma \vdash M:A$ be   $\forall$-lazy and cut-free, then:
		\begin{itemize}
		\item  $\mathcal{D}$ contains no instances of $\with$R$1$, $\with$L$i$ and  $\forall$L;
		\item  $M$ is normal and contains  no occurrences of  $\mathtt{copy}$ and $\pi_i$; 
		\item  if $\Gamma=\emptyset$ then $M\in \mathcal{V}$.
		\end{itemize}
	\end{enumerate}
\end{prop}

\begin{prop}[$\eta$-expanded derivations]\label{prop: properties of eta logn nf} Let $\mathcal{D}\triangleleft \Gamma \vdash M:A$ be $\forall$-lazy and $\eta$-expanded. Then:
	\begin{enumerate}
		\item  \label{point: properties eta long nf 1}   $\vert M \vert \leq \vert \Gamma \vert  + \vert A \vert \leq 2 \cdot \vert \mathcal{D}\vert$;
		\item \label{point: properties eta long nf 2}  if  $\mathcal{D}'\triangleleft \Gamma \vdash N:A$ is  cut-free  for some $N\in \Lambda_{\mathtt{ copy}}$, then both $\vert N \vert \leq \s{M}$ and $\s{\mathcal{D}'}\leq \s{\mathcal{D}}$.
	\end{enumerate}
\end{prop}
\begin{proof}
	By Definition~\ref{defn: eta long nf}  and Proposition~\ref{prop: lazyness propagation}.\ref{point 2: laziness propagation},   $\mathcal{D}$ is cut-free and without  instances of $\with$R$1$, $\with$L$i$,  $\forall$L,  with $M$  normal and free from  $\mathtt{copy}$ constructs. Point~\ref{point: properties eta long nf 1} is a straightforward  induction on $\mathcal{D}$.
	Concerning  point~\ref{point: properties eta long nf 2}, we have $\vert \mathcal{D}'\vert  \leq \s{\mathcal{D}}$, because $\mathcal{D}$ is $\eta$-expanded. Now, notice that $\s{N}+ \forall (\Gamma, A)= \s{\mathcal{D}'}+ \forall_{ax}(\mathcal{D}')$, where $\forall(\Gamma, A)$  denotes  the number of occurrences of $\forall$   in $\Gamma, A$, and $\forall_{ax}(\mathcal{D}')$ denotes  the number of occurrences of $\forall$ in the conclusions of the instances of  $ax$  in $\mathcal{D}'$.  We have, $\s{N}+ \forall (\Gamma, A)= \s{\mathcal{D}'}+ \forall_{ax}(\mathcal{D}') \leq \s{\mathcal{D}}+ \forall_{ax}(\mathcal{D}) = \s{M}+ \forall (\Gamma, A)$,  which implies $\vert N \vert \leq \vert M \vert$.
\end{proof}
\begin{rem}\label{rem: eta long nf} Given $\mathcal{D}\triangleleft \Gamma \vdash M:A$ a $\forall$-lazy and cut-free derivation, we can always construct a $\eta$-expanded derivation $\mathcal{D}^* \triangleleft \Gamma \vdash M^*:A$. Moreover,  Proposition~\ref{prop: properties of eta logn nf}.\ref{point: properties eta long nf 2} implies  both $\vert N \vert \leq \vert M^* \vert$ and $\vert \mathcal{D}' \vert\leq \vert \mathcal{D}^* \vert$, for any cut-free derivation $\mathcal{D}' \triangleleft \Gamma \vdash N:A$.  Henceforth, w.l.o.g.~we shall assume that the derivation of the premise $ \vdash V:A$ of $\with$R$1$ is $\eta$-expanded. 
\end{rem}

Remark~\ref{rem: eta long nf}  prevents the increase of size   during  normalization, since  the size of $V$ in $\with$R$1$ is always  bigger than then size of any value $U$   of type $A$, and so the construct $\mathtt{ copy}^V$ bounds the size of the new copy of $U$  produced  by applying the reduction rule in~\eqref{eqn: reduction rules LAML}. 

\subsection{Linear Additives Prevent the Exponential Blowup}\label{subsection: linear additives prevent blow up}
We can observe the benefits of moving from the  standard additive rules to the linear additive rules as soon as we adapt the constructions in   Definition~\ref{defn: nesting pairs} to $\mathsf{LAM}$.
\begin{defn}[Nesting $\mathsf{LAM}$ terms]\label{defn: nesting pairs LAM} Let  $V \in \mathcal{V}$ and $k \in \mathbb{N}$. For all $n \in \mathbb{N}$ we define $\mathtt{ladd}^{x, V_{[k]}}_{n}$ as follows:
	\begin{equation*}
	\mathtt{ladd}^{x, V_{[k]}}_{n}\triangleq \begin{cases}x &n=0\\
	(\lambda y. \mathtt{ladd}^{y, V_{[k+1]}}_{n-1}) (\mathtt{copy}^{V_{[k]}} x\  \mathtt{as}\ x_1, x_2 \mathtt{\ in \ }\langle x_1,x_2\rangle )&n>0
	\end{cases} 
	\end{equation*}
	where $V_{[k]}$ is as in~\eqref{eqn: type An and term Mn}.
\end{defn} 

The following proposition states that nesting instances of $\with$R$1$ in a derivation of $\mathsf{LAM}$ produces no exponential blow up.
\begin{prop}[Linearity of $\mathsf{LAM}$]\label{prop: linearity of LAM} Let $A$ be a closed and $\forall$-lazy type, and let $U, V\in \mathcal{V}$ be   inhabitants of $A$, with $V$ a $\eta$-long normal form. For all $n \in \mathbb{N}$:
	\begin{enumerate}
		\item \label{enum: size explosion 1} $\vdash \lambda x.\mathtt{ladd}^{x, V}_{n} :A\multimap A_{[n]}$ is derivable in $\mathsf{LAM}$;
		\item \label{enum: size explosion 4} $(\lambda x. \mathtt{ladd}^{x, V}_n)U\rightarrow^{2n+1}  U_{[n]}$, where  $\vert (\lambda x. \mathtt{ladd}^{x, V}_n)U\vert > \vert U_{[n]} \vert>2n+1$.
	\end{enumerate}
\end{prop}
\begin{proof}
	Concerning point~\ref{enum: size explosion 1}, we prove by induction on $n $ that $\lambda x.\mathtt{ladd}^{x, V_{[k]}}_{n}$ has type  $A_{[k]} \multimap A_{[k + n]}$, for all $k \in \mathbb{N}$. If $n=0$ then  $\mathtt{ladd}^{x, V_{[k]}}_{n}=x$, so that $\lambda x.x$ has type $A_{[k]}\multimap A_{[k]}$. Let us consider the case  $n>0$.  By induction hypothesis,   $\lambda y.\mathtt{ladd}^{y, V_{[k+1]}}_{n-1}$ has type $A_{[k+1]}\multimap A_{[k+n]}$. If $x,x_1,x_2$ have type $A_{[k]}$ then  $\mathtt{copy}^{V_{[k]}} x\  \mathtt{as}\ x_1, x_2 \mathtt{\ in \ }\langle x_1,x_2\rangle $ has type $A_{[k+1]}$ and  $(\lambda y.\mathtt{ladd}^{y, V_{[k+1]}}_{n-1})(\mathtt{copy}^{V_{[k]}} x\  \mathtt{as}\ x_1, x_2 \mathtt{\ in \ }\langle x_1,x_2\rangle )$ has type $A_{[k+ n]}$. Therefore,   $\lambda x.\mathtt{ladd}^{x, V_{[k]}}_{n}$ has type $A_{[k]} \multimap A_{[k+n]}$.  Concerning point~\ref{enum: size explosion 2}, we show  by induction on $n $ that $(\mathtt{ladd}^{x, V_{[k]}}_{n})[U_{[k]}/x]\rightarrow^{2n} U_{[n+k]}$, for all $k \in \mathbb{N}$. The base case is trivial. If $n>0$ then, since $U_{ [k]}\in \mathcal{V}$, we have: 
	\begin{equation*}
	\begin{aligned}
	(\mathtt{ladd}^{x, V_{[k]}}_{n})[U_{[k]}/x]&= (\lambda y. \mathtt{ladd}^{y, V_{[k+1]}}_{n-1})(\mathtt{copy}^{V_{[k]}} U_{ [k]} \mathtt{as}\ x_1, x_2 \mathtt{\ in \ }\langle x_1,x_2\rangle ) \\
	&\rightarrow  (\lambda y. \mathtt{ladd}^{y, V_{[k+1]}}_{n-1})\langle U_{[k]},U_{[k]}\rangle \\
	&\rightarrow \mathtt{ladd}^{y, V_{[k+1]}}_{n-1}[\langle U_{[k]},U_{[k]}\rangle /y]=\mathtt{ladd}^{y, V_{[k+1]}}_{n-1}[U_{[k+1]}/y]
	\end{aligned}
	\end{equation*}
	which reduces in $2(n-1)$ steps to $U_{[k+n]}$ by induction hypothesis.  Finally,  we notice   that  $\vert   \mathtt{ladd}^{x, V_{[k]}}_n \vert = 7n+   \sum^{n-1}_{i=0} \vert V_{[k+i]} \vert + 1$ and  $\vert U_{[n]} \vert = 2^{n \cdot \vert U \vert} + \sum^{n-1}_{i=0}2^{i}$, for all $n \in \mathbb{N}$. Since  $V$ is  a $\eta$-long normal form of $A$, we have  $\vert V \vert \geq \vert U \vert$ by  Proposition~\ref{prop: properties of eta logn nf}.\ref{point: properties eta long nf 2}, and the conclusion follows.
\end{proof}


\section{Computational and Proof-Theoretical Properties of $\mathsf{LAM}$}
In Section~\ref{subsection: linear additives prevent blow up} we have shown with the help of  a key example how linear additives  prevent exponential explosions in normalization. We now investigate further this point  by proving that  $\mathsf{LAM}$ enjoys both linear strong normalization (Theorem~\ref{thm:lazy normalization for LAML}) and a mildly weakened cubic cut-elimination property (Theorem~\ref{thm: cut elimination for LAM}). 

\subsection{Subject Reduction and  Linear Strong  Normalization}
Linear strong normalization for $\mathsf{LAM}$ is achieved by proving  that reduction preserves types and  shrinks the size of typable terms, i.e.~a slightly stronger version of Subject reduction. First, we need some straightforward preliminary lemmas:
\begin{lem}[Linearity] \label{lem: linearity}If $\mathcal{D}\triangleleft \Gamma, x:A \vdash M:C$ then $x$ occurs exactly once in $M$.
\end{lem}
\begin{lem}[Generation] \label{lem: generation}  {\ }  
	\begin{enumerate}
		\item \label{lem:generation 1} 
		If $\mathcal{D} \triangleleft \Gamma \vdash \lambda x. M:A$ then 
		$A=  \forall \alpha_1 \ldots \forall \alpha_n . (B \multimap C)$  and, by permuting some rules of $\mathcal{D}$, we obtain a derivation $\mathcal{D}'$ of $\Gamma, x:B\vdash M: C$, followed by an instance of $\multimap$R and a sequence of $\forall$R.
		\item If $\mathcal{D} \triangleleft \Gamma, x: \forall \alpha. A \vdash M[xN/y]:B$ then $\mathcal{D}$ contains an instance of $\forall$L that introduces $  x: \forall \alpha. A$.
		\item If $\mathcal{D} \triangleleft \Gamma, x:  A \multimap B\vdash M[xN/y]:C$ then $\mathcal{D}$ contains an instance of $\multimap$L that introduces $x: A \multimap B$.
		\item \label{lem: generation withR0} If $\mathcal{D}\triangleleft \Gamma   \vdash \langle M_1, M_2\rangle :  A$ then $\Gamma=\emptyset$,  $A= B_1 \with B_2$, and the last rule of $\mathcal{D}$ is $\with$R$0$.
		\item \label{enum: generation for LAML pi} If $\mathcal{D}\triangleleft \Gamma, x:A_1\with A_2 \vdash M[\pi_{i}(x)/x_i ]:A$ then $\mathcal{D}$ contains an instance of  $\with$L$i$ that introduces $x:A_1\with A_2$.
		\item \label{lem: generation withR1} If $\mathcal{D}\triangleleft \Gamma, x:B \vdash M[\mathtt{copy}^V  x \mathtt{\ as \ }x_1,x_2 \mathtt{\ in \ } \langle N_1, N_2\rangle /y]:  A$, then   $\mathcal{D}$ contains an instance of $\with$R$1$ that introduces $x:B$.
	\end{enumerate}
\end{lem}

\begin{thm}[Subject reduction] 
	\label{thm: Subject Reduction for LAM} 
	If  $\mathcal{D}\triangleleft \Gamma \vdash  M_1: A$ and  $M_1 \rightarrow M_2$ then:
	\begin{itemize}
		\item $\s{M_2}<\s{M_1}$, and
		\item $\mathcal{D}^*\triangleleft \Gamma \vdash M_2:A$, for some  $\mathcal{D}^*$.   
	\end{itemize}  
\end{thm}
\begin{proof}
	We proceed by structural induction on $ \mathcal{D} $. The crucial case is when the last rule of $\mathcal{D}$ is an instance of  $cut$ introducing the redex in $M_1$ that has been fired by the reduction step $M_1 \rightarrow M_2$.  We just consider the case where $ M_1= P[\mathtt{copy}^{U} V   \mathtt{\ as \ }x_1,x_2 \mathtt{\ in \ } \langle N_1, N_2\rangle /y]$ and $M_2= P[\langle N_1[V/x_1], N_2[V/x_2]\rangle /y]$. We have:
	\begin{prooftree}
		\AxiomC{$\mathcal{D}_1  \triangleleft  \vdash V: B$}
		\AxiomC{$\mathcal{D}_2 \triangleleft \Delta, x:  B \vdash  P[\mathtt{copy}^{U}  x \mathtt{\ as \ }x_1,x_2 \mathtt{\ in \ } \langle N_1, N_2\rangle /y]:A$}
		\RightLabel{$ cut $  }
		\BinaryInfC{$\mathcal{D} \triangleleft  \Delta \vdash P[ \mathtt{copy}^{U}  V \mathtt{\ as \ }x_1,x_2 \mathtt{\ in \ } \langle N_1, N_2\rangle /y]:A$}
	\end{prooftree}
	By applying Lemma~\ref{lem: generation}.\ref{lem: generation withR1},  $\mathcal{D}_2$ must be of the following form:
	\begin{prooftree}
		\def\defaultHypSeparation{\hskip .1cm}
		\AxiomC{$\mathcal{D}'\triangleleft x_1:B \vdash N_1: B_1$}
		\AxiomC{$\mathcal{D}'' \triangleleft x_2:B \vdash N_2: B_2$}
		\AxiomC{$\mathcal{D}''' \triangleleft \vdash U: B$}
		\RightLabel{$\with$R$1$}
		\TrinaryInfC{$x: B \vdash \mathtt{copy}^{U} x \mathtt{\ as \ }x_1, x_2 \mathtt{\ in \ }\langle N_1, N_2\rangle: B_1 \with B_2$}
		\noLine
		\UnaryInfC{$\vdots \ \gamma$}
		\noLine
		\UnaryInfC{$\mathcal{D}_2 \triangleleft \Delta, x:  B \vdash  P[\mathtt{copy}^{U}  x \mathtt{\ as \ }x_1,x_2 \mathtt{\ in \ } \langle N_1, N_2\rangle /y]:A$}
	\end{prooftree}
	where $\gamma$ is a sequence of rules. We construct $\mathcal{D}^*$  as the following derivation:
	\begin{prooftree}
		\def\defaultHypSeparation{\hskip .1cm}
		\AxiomC{$\mathcal{D}_1\triangleleft \vdash V:B$}
		\AxiomC{$\mathcal{D}'\triangleleft x_1:B \vdash N_1: B_1$}
		\RightLabel{$cut$}
		\BinaryInfC{$\vdash N_1[V/x_1]:B_1$}
		\AxiomC{$\mathcal{D}_1\triangleleft \vdash V:B$}
		\AxiomC{$\mathcal{D}''\triangleleft x_2:B \vdash N_2: B_2$}
		\RightLabel{$cut$}
		\BinaryInfC{$\vdash N_2[V/x_2]:B_2$}
		\RightLabel{$\with$R$0$}
		\BinaryInfC{$\vdash \langle N_1[V/x_1], N_2[V/x_2]\rangle : B_1 \with B_2$}
		\noLine
		\UnaryInfC{$\vdots \ \gamma$}
		\noLine
		\UnaryInfC{$ \Delta \vdash  P[\langle N_1[V/x_1], N_2[V/x_2]\rangle /y]:A$}
	\end{prooftree}
	By Remark~\ref{rem: eta long nf}, $U$ is a $\eta$-long normal form, so that $\s{V}\leq \s{U}$ by Proposition~\ref{prop: properties of eta logn nf}.\ref{point: properties eta long nf 2}.  By Lemma~\ref{lem: linearity}, $x_1$ and $x_2$ occur exactly once in $N_1$ and $N_2$, respectively. Hence, $ \s{N_i[V/x_i]}= \s{N_i}+ \s{V}-1$, for $i=1,2$.     We have:
	\begin{equation*}
	\begin{split}
	\s{\langle N_1[V/x_1], N_2[V/x_2]\rangle }&= \s{N_1[V/x_1]}+ \s{N_2[V/x_2]}+1 = 2\cdot \s{V} +  \s{N_1}+\s{N_2} -1\\
	&  \leq  \s{U}+ \s{V}+   \s{N_1}+\s{N_2} -1  <   \s{\mathtt{copy}^{U}  V \mathtt{\ as \ }x_1,x_2 \mathtt{\ in \ } \langle N_1, N_2\rangle }  
	\end{split}
	\end{equation*}
	and this implies $\s{M_2}<\s{M_1}$.
\end{proof}

Subject reduction entails linear strong normalization.
\begin{thm}[Linear strong normalization]\label{thm:lazy normalization for LAML} If $\mathcal{D}\triangleleft \Gamma \vdash M: A$ then $M$ reduces  to a  normal  form in at most $\s{M}$ reduction steps. 
\end{thm}

\begin{rem}\label{rem: uniqueness nf}
	By Newman's Lemma (see~\cite{barendregt2013lambda}),  confluence of $\rightarrow$ for typable terms holds as a consequence of Theorem~\ref{thm:lazy normalization for LAML} and weak confluence, the latter being easy to establish. Therefore, each typable term  reduces to a \emph{unique} normal form.
\end{rem}

\subsection{Cubic $\forall$-Lazy Cut-Elimination}
$\mathsf{LAM}$ is a subsystem of $\mathsf{IMALL}_2$. Hence, 
 from a purely proof-theoretical viewpoint, the former inherits the  cut-elimination rules of the latter (see for example~\cite{brauner1996cut}). However, these proof rewriting rules for $\mathsf{LAM}$ would be more permissive than the  reduction rules for $\Lambda_{\mathtt{copy}}$, essentially because the  $\mathtt{ copy}$ construct can only duplicate values. The next examples illustrate this point.

\begin{exmp} Let us consider the following derivation of $\mathsf{LAM}$:
\begin{equation}\label{eqn: deadlocks in LAM}
\AxiomC{$y: \mathbf{1} \vdash y \mathbf{I}: \mathbf{1}$}
\AxiomC{$x_1: \mathbf{1}\vdash x_1: \mathbf{1}$}
\AxiomC{$x_2: \mathbf{1}\vdash x_2: \mathbf{1}$}
\AxiomC{$\vdash \mathbf{I}: \mathbf{1}$}
\RightLabel{$\with$R$1$  }
\TrinaryInfC{$ x: \mathbf{1}  \vdash \mathtt{copy} ^{\mathbf I}\, x \mathtt{ \   as \ } x_1,x_2 \mathtt{ \   in \ } \langle x_1, x_2 \rangle : \mathbf{1}\with  \mathbf{1} $}
\RightLabel{$cut$}
\BinaryInfC{$ y: \mathbf{1}  \vdash \mathtt{copy} ^{\mathbf{I}} \, y\mathbf{I} \mathtt{ \   as \ } x_1,x_2 \mathtt{ \   in \ } \langle  x_1, x_2 \rangle : \mathbf{1}\with  \mathbf{1} $}
\DisplayProof
\end{equation}
and let us apply the cut-elimination rule of $\mathsf{IMALL}_2$ moving the $cut$ upward. We get a derivation of $y: \mathbf{1}  \vdash  M: \mathbf{1}\with \mathbf{1}$, where $M$ is $ \mathtt{copy} ^{\mathbf{I}} \, y \mathtt{ \   as \ } y_1,y_2 \mathtt{ \   in \ } \langle y_1 \mathbf{I}, y_2 \mathbf{I}\rangle $. But $  \mathtt{copy} ^{\mathbf{I}} \, y\mathbf{I} \mathtt{ \   as \ } x_1,x_2 \mathtt{ \   in \ } \langle  x_1, x_2 \rangle \not \rightarrow^*M$.
\end{exmp} 
\noindent

\begin{exmp} Let us consider the following derivation of $\mathsf{LAM}$:
	\begin{equation}\label{eqn: copy-first in LAM}
	\AxiomC{$x_1: \mathbf{1}\vdash x_1: \mathbf{1}$}
	\AxiomC{$x_2: \mathbf{1}\vdash x_2: \mathbf{1}$}
	\AxiomC{$\vdash \mathbf{I}: \mathbf{1}$}
	\RightLabel{$\with$R$1$ }
	\TrinaryInfC{$ x: \mathbf{1}  \vdash \mathtt{copy} ^{\mathbf I}\, x \mathtt{ \   as \ } x_1,x_2 \mathtt{ \   in \ } \langle x_1, x_2 \rangle : \mathbf{1}\with  \mathbf{1} $}
	\AxiomC{$z: \mathbf{1}\vdash z: \mathbf{1}$}
	\RightLabel{$\with$L$i$}
	\UnaryInfC{$y: \mathbf{1}\with \mathbf{1}\vdash \pi_1(y): \mathbf{1}$}
	\RightLabel{$cut$}
	\BinaryInfC{$ x: \mathbf{1}  \vdash \pi_1(\mathtt{copy} ^{\mathbf{I}} \, x \mathtt{ \   as \ } x_1,x_2 \mathtt{ \   in \ } \langle  x_1, x_2 \rangle ): \mathbf{1}$}
	\DisplayProof
	\end{equation}
	and let us apply the principal cut-elimination rule  for $\with$ in $\mathsf{IMALL}_2$. We get a \textit{cut} with premises $x_1: \mathbf{1}\vdash x_1: \mathbf{1}$ and $z: \mathbf{1}\vdash z: \mathbf{1}$. However, no reduction rule rewrites $\pi_1(\mathtt{copy} ^{\mathbf{I}} \, x \mathtt{ \   as \ } x_1,x_2 \mathtt{ \   in \ } \langle  x_1, x_2 \rangle )$ into $x_1$.
\end{exmp}

To circumvent the above mismatch between proof rewriting  and term reduction, we introduce the  \emph{$\forall$-lazy cut-elimination rules}. They  never eliminate  instances of $cut$ like~\eqref{eqn: deadlocks in LAM} and~\eqref{eqn: copy-first in LAM}, so that cut-elimination fails in general. We then show that, by defining a special $\forall$-lazy cut-elimination strategy with cubic cost, cut-elimination can be recovered in the restricted case of derivations of $\forall$-lazy types (Theorem~\ref{thm: cut elimination for LAM}). This result is analogous to the restricted (lazy) cut-elimination property for derivations of lazy types  discussed in Section~\ref{subsec:blowuplazy}. The crucial difference is that, while Girard's lazy types rule out both negative occurrences of $\forall$ and  positive occurrences of $\with$,  the  $\forall$-lazy types only require the absence of negative $\forall$, as illustrated in Figure~\ref{fig: foralllazy}. This allows  us to nest instances of $\with$R$1$ in a  derivation of $\mathsf{LAM}$ without incurring  exponential proof normalization, as shown in Section~\ref{subsection: linear additives prevent blow up} in the case of  term reduction.

\begin{figure}[t]
	\centering
	\begin{tabular}{|c|c|c|}
		\cline{2-3}
		\multicolumn{1}{c|}{}&\multicolumn{1}{c|}{\textit{negative} $\forall$}&\multicolumn{1}{c|}{\textit{positive} $\with$}\\ \hline
		\textit{lazy types} &   $X$ &    $X$   \\ \hline
		\textit{$\forall$-lazy types} &  $X$ &  $\checkmark$ \\ \hline
	\end{tabular}
	\caption{Lazy types vs $\forall$-lazy types.}
\label{fig: foralllazy}
\end{figure}

\begin{defn}[$\forall$-lazy cut-elimination rules]{\ }
\label{defn:Lazy cut-elimination rules}
\begin{itemize}
\item  With $ (X,Y) $ we denote a $cut$ whose left (resp.~right) premise is the conclusion of an instance of the inference rule  $X$ (resp.~$Y$). Cuts are divided into four classes: the \emph{symmetric cuts} are $(\multimap\text{R},\multimap\text{L})$,  $(\with\text{R}0,\with\text{L}i)$,  $(\forall \text{R},\forall\text{L})$ and those of the form $(X,ax)$ or $(ax,Y)$, for some $ X$ and $ Y $; the \emph{copy-first cuts}  have form $(\with{R}1,\with\text{L}i )$;    the  \textit{critical cuts} have  form $(X, \with\text{R}1)$, for some rule $X$; finally, the  \emph{commuting cuts} are all the remaining instances of $ cut $.
\item Let the following be a critical cut:
\begin{prooftree}
	\AxiomC{$\mathcal{D}$}
	\noLine
\UnaryInfC{$\Gamma \vdash M: A$}
\AxiomC{$x_1: A\vdash N_1: B$}
\AxiomC{$x_2: A\vdash N_2: C$}
\AxiomC{$\vdash V: A$}
\RightLabel{$\with$R$1$}
\TrinaryInfC{$ x: A \vdash \mathtt{copy} ^{V}\, x \mathtt{ \   as \ } x_1,x_2 \mathtt{ \   in \ } \langle N_1, N_2 \rangle : B\with C $}
\RightLabel{$cut$}
\BinaryInfC{$ \Gamma  \vdash \mathtt{copy} ^{V} \, M \mathtt{ \   as \ } x_1,x_2 \mathtt{ \   in \ } \langle  N_1,N_2 \rangle :B\with  C$}
\end{prooftree}
It is called  \textit{safe} if $\Gamma=\emptyset$, and  \textit{deadlock} otherwise. Also, it is called  \textit{ready} if it is safe and $\mathcal{D}$ is cut-free. In this case, since $A$ is $\forall$-lazy,  Proposition~\ref{prop: lazyness propagation}.\ref{point 2: laziness propagation} implies $M \in \mathcal{V}$.
\item The \emph{$\forall$-lazy cut-elimination rules} are defined as follows:
\begin{itemize}
\item they correspond to the standard cut-elimination rules of $\mathsf{LL}$ for commuting cuts and  for the symmetric cuts $(\multimap\text{R},\multimap\text{L})$,  $(\forall \text{R},\forall\text{L})$, $(X,ax)$ and $(ax,Y)$ (see e.g.~\cite{brauner1996cut}); 
\item they are displayed in Figure~\ref{fig: cut elimination LAM} for the symmetric cut   $(\with\text{R}0,\with\text{L}i)$ and for those critical cuts $(X, \with\text{R}1)$ which are \textit{ready};
\item there is no $\forall$-lazy cut-elimination rule for copy-first cuts and the remaining critical cuts.
\end{itemize}
If $\mathcal{D}$ rewrites  to  $\mathcal{D}'$ by 
 a $\forall$-lazy cut-elimination rule, we write $\mathcal{D}\red \mathcal{D}'$.  The reflexive and transitive closure of $\red$ is $\red^*$. 
\end{itemize}
\end{defn}

The $\forall$-lazy cut-elimination rules prevent  duplication of terms that are not values,  restoring a match between proof rewriting and term reduction. What we are  going to show is that any $\forall$-lazy derivation $\mathcal{D}$ can  be rewritten into a cut-free one by a specific strategy of $\forall$-lazy cut-elimination steps. This implies that any instance of critical cut in $\mathcal{D}$ (e.g. the deadlock  in~\eqref{eqn: deadlocks in LAM}) is eventually turned into a ready cut, and that all instances of copy-first cuts like~\eqref{eqn: copy-first in LAM}  are eventually  turned into   $(\with\text{R}0, \with\text{L}i)$ cuts. The latter is due to the replacement of $\with$R$1$ with  $\with$R$0$ given by the  $\forall$-lazy cut elimination of ready cuts in Figure~\ref{fig: cut elimination LAM}.

\begin{figure}[t]
\begin{framed}
\scalebox{0.7}{$
\begin{aligned}
&(\with \text{R}0, \with\text{L}i) & & 
\def\ScoreOverhang{1pt}
	\AxiomC{$\vdash N_1: A_1$}
	\AxiomC{$\vdash N_2: A_2$}
	\BinaryInfC{$\vdash \langle N_1, N_2\rangle : A_1\with A_2$}
	\AxiomC{$\Gamma, x_i:A_i \vdash M: B$}
	\AxiomC{$i \in \{ 1,2\}$}
	\BinaryInfC{$\Gamma, x: A_1\with A_2\vdash M[\pi_i (x)/x_i]: B$}
	\RightLabel{$cut$}
	\BinaryInfC{$\Gamma \vdash M[\pi_i \langle N_1, N_2\rangle /x_i]: B$}
	\DisplayProof
 \red \ 
\def\ScoreOverhang{1pt}
\AxiomC{$\vdash N_i: A_i$}
	\AxiomC{$\Gamma, x_i:A_i \vdash M: B$}
	\RightLabel{$cut$}
	\BinaryInfC{$\Gamma \vdash M[N_i/x_i]: B$}
	\DisplayProof
\\
\\
&(X, \with \text{R}1) & & 
	\def\ScoreOverhang{1pt}
\AxiomC{$\mathcal{D}^\dagger$}
\noLine
\UnaryInfC{$\vdash  V: A$}
\AxiomC{$x_1: A \vdash N_1: B_1$}
\AxiomC{$x_2: A \vdash N_2: B_2$}
\AxiomC{$\vdash U:A$}
\TrinaryInfC{$x:A \vdash \mathtt{copy}^{U}  x \mathtt{\ as\ } x_1, x_2 \mathtt{\ in\ }\langle N_1, N_2\rangle :B_1 \with B_2$}
\RightLabel{$cut$}
\BinaryInfC{$ \vdash \mathtt{copy}^{U} V
 \mathtt{\ as\ } x_1, x_2 \mathtt{\ in \ }\langle N_1, N_2\rangle :B_1 \with B_2$}
\DisplayProof
\red \ 
\def\ScoreOverhang{1pt}
\AxiomC{$\mathcal{D}^\dagger$}
\noLine
\UnaryInfC{$\vdash  V: A$}
\AxiomC{$x_1: A \vdash N_1: B_1$}
\RightLabel{$cut$}
\BinaryInfC{$\vdash N_1[V/x_1]:B_1$}
\AxiomC{$\mathcal{D}^\dagger$}
\noLine
\UnaryInfC{$\vdash  V: A$}
\AxiomC{$x_2: A \vdash N_2: B_1$}
\RightLabel{$cut$}
\BinaryInfC{$\vdash N_2[V/x_2]:B_2$}
\RightLabel{$\with$R$0$}
\BinaryInfC{$\vdash \langle  N_1[V/x_1], N_2[V/x_2]\rangle :B_1 \with B_2$}
\DisplayProof
\end{aligned}$}
\caption{$\forall$-lazy cut-elimination rules for $(\with\text{R}0, \with\text{L}i)$ and for  ready cuts,  where $\mathcal{D^\dagger}$ is cut-free.}
\label{fig: cut elimination LAM} 
\end{framed}
\end{figure}

The  proposition below exploits the  \textit{closure} conditions introduced in  Definition~\ref{defn: the system LAM}.  
\begin{prop} \label{prop: closed context from closed conclusion}If $\mathcal{D}\triangleleft \Gamma \vdash M:A$ and $FV(A)= \emptyset$ then $FV(\Gamma)= \emptyset$. 
\end{prop}
\begin{proof}
By induction on $\mathcal{D}$ using the closure conditions (i)-(ii) in Definition~\ref{defn: the system LAM} and the conditions on linear additives. We only consider some interesting cases.
 Suppose $\mathcal{D}$ ends with an instance of $\multimap$L with premises $\Gamma' \vdash N:B$   and $\Delta, x:C \vdash M':A$. If $FV(A)=\emptyset$ then, by induction hypothesis, 
 $FV(\Delta)= FV(C)=\emptyset$. By applying the closure condition (i), we have  $FV(B)=\emptyset$. By applying, again, the induction hypothesis we have $FV(\Gamma')=\emptyset$, and hence   $FV(\Gamma',  y: B \multimap C, \Delta)=\emptyset$.
  Let us now consider the case where $\mathcal{D}$ ends with an instance of $\with$L$i$ with premise $\Gamma', x_i: B_i \vdash M':A$ and conclusion $\Gamma',  y :B_1 \with B_2  \vdash M'[\pi_i(y)/x_i] :A
   $. If $FV(A)=\emptyset$ then $FV(\Gamma')=\emptyset$ by induction hypothesis. Moreover, since $B_1\with B_2$ is a closed $\forall$-lazy type, we conclude  $FV(\Gamma', y :B_1 \with B_2 )= \emptyset$.
\end{proof}
The next  lemmas are essential  to ensure  the restricted  cut-elimination result for $\forall$-lazy types.
\begin{lem}\label{lem: intermediate lemma}
Let $\mathcal{D}\triangleleft \Gamma \vdash M: A$ be a derivation whose only cuts are either deadlocks or copy-first. If one of those cuts exists in $\mathcal{D}$, then $\mathcal{D}$ is not $\forall$-lazy.
\end{lem}
\begin{proof}
 First, we show that both the conclusion of a deadlock  and  the conclusion of a copy-first cut  cannot be   $\forall$-lazy.  It suffices to find a closed type in the context of these judgments, since closed types must contain at least a $\forall$ in positive position. Let 
\begin{prooftree}
	\AxiomC{$\Delta \vdash N: B$}
	\AxiomC{$ x:B\vdash  \mathtt{copy}^{V}  x \mathtt{\ as \ }x_1, x_2 \mathtt{ \ in\ } \langle P,Q \rangle:C \with D$ }
	\RightLabel{$cut$}
	\BinaryInfC{$ \Delta \vdash  \mathtt{copy}^{V}  N \mathtt{\ as \ }x_1, x_2 \mathtt{ \ in\ } \langle P,Q \rangle:C \with D$}
\end{prooftree}
be a deadlock. By definition, we have $\Delta \neq \emptyset$. Since $ x:B\vdash  \mathtt{copy}^{V}  x \mathtt{\ as \ }x_1, x_2 \mathtt{ \ in\ } \langle N_1,N_2 \rangle:C \with D$ is the conclusion of $\with$R$1$, we have $FV(B)=\emptyset$. Hence  $FV(\Delta)=\emptyset$, by Proposition~\ref{prop: closed context from closed conclusion}.  Moreover, let  
\begin{prooftree}
	\AxiomC{$x:B \vdash  \mathtt{copy}^{V}  x \mathtt{\ as \ }x_1, x_2 \mathtt{ \ in\ } \langle N_1,N_2 \rangle: B_1 \with B_2$}
	\AxiomC{$\Delta, x: B_1\with B_2\vdash M[\pi_i (x)/x_i]: C$}
	\RightLabel{$cut$}
	\BinaryInfC{$\Delta, x:B \vdash M[\pi_i ( \mathtt{copy}^{V}  x \mathtt{\ as \ }x_1, x_2 \mathtt{ \ in\ } \langle N_1,N_2 \rangle) /x_i]: C$}
\end{prooftree}
be a copy-first cut. Its leftmost premise is the conclusion of  $\with$R$1$, so $B$ must be closed by definition. 

Suppose now that $\mathcal{D}$ contains some cuts. Then it contains at least a deadlock or a copy-first cut. In both cases, $\mathcal{D}$ contains a judgment that is not $\forall$-lazy. Let $ R_1,\ldots R_n $ be the sequence of rule instances from  $\Gamma \vdash M: A$ up to this judgment. We prove by induction on $n$ that   $ \Gamma \vdash M: A$ cannot be  $\forall$-lazy. The case $n=0$ is trivial.  If $n>0$, then we have two cases depending on $R_n$. If $R_n$ is a cut, then it is either a deadlock or a copy-first cut,  and its conclusion cannot be $\forall$-lazy, so we  apply the induction hypothesis. If $R_n$ is not a  cut,   we apply Proposition~\ref{prop: lazyness propagation}.\ref{point 1: laziness propagation} and the induction hypothesis. 
\end{proof}
\begin{lem}[Existence of a safe cut] 
\label{lem: existence of the safe exponential cut} 
Let $\mathcal{D}\triangleleft \Gamma \vdash M: A$ be a $\forall$-lazy derivation whose only cuts are either critical or copy-first. Then:
\begin{enumerate}
\item\label{enum: safe} if $\mathcal{D}$ has critical cuts, then it  has safe cuts;
\item \label{enum: copy-first} if $\mathcal{D}$ is free from critical cuts, then it is free from copy-first cuts. 
\end{enumerate}
\end{lem}
\begin{proof}
 Both points follow from Lemma~\ref{lem: intermediate lemma}.
\end{proof}

\begin{defn}[Height and weight]
\label{defn:Height of an inference rule} Let $\mathcal{D}\triangleleft \Gamma \vdash M:A$ be a derivation of $\mathsf{LAM}$:
\begin{itemize}
\item  The \textit{weight of $\mathcal{D}$}, written $\with(\mathcal{D})$, is the number of instances of the rule $\with$R$1$ in $\mathcal{D}$.
\item Given a rule instance $R$ in $\mathcal{D}$, the  \emph{height of $R$}, written $ h(R) $,  is the number of rule instances from the conclusion $ \Gamma\vdash M:A $ of $ \mathcal{D} $ upward to the conclusion of $ R $.  The  \textit{height of $\mathcal{D}$}, written $h(\mathcal{D})$, is the largest $h(R)$ among its rule instances.
\end{itemize}
\end{defn}

\begin{lem}[Eliminating a ready cut]
\label{lem: eliminating a lazy cut}
 Let $\mathcal{D}\triangleleft \Gamma \vdash M: A$ be a $\forall$-lazy derivation whose only cuts are either critical or  copy-first. The following statements hold:
 \begin{enumerate}
 \item \label{enum: eliminating ready cut 1} if $\mathcal{D} $  has critical cuts,  then it has ready cuts; 
 \item \label{enum: eliminating ready cut 2} if $\mathcal{D}^*$ is obtained by eliminating  a ready cut in $\mathcal{D}$, then  $\vert\mathcal{D}^*\vert +  2\cdot \with(\mathcal{D}^*)<  \vert\mathcal{D}\vert+2\cdot \with(\mathcal{D})$.
 \end{enumerate}
\end{lem}
\begin{proof}
Concerning point~\ref{enum: eliminating ready cut 1},  by Lemma~\ref{lem: existence of the safe exponential cut}.\ref{enum: safe} 
$ \mathcal{D} $ contains at least one  safe cut. 
Let $R$ be the one with maximal height, we display as follows:
\begin{equation}
\label{eqn: above the cut}
\AxiomC{$\mathcal{D}'$}
\noLine
\UnaryInfC{$\vdash  N: B$}
\AxiomC{$\mathcal{D}_1$}
\noLine
\UnaryInfC{$x_1: B \vdash M_1: B_1$}
\AxiomC{$\mathcal{D}_2$}
\noLine
\UnaryInfC{$x_2: B \vdash M_2: B_2$}
\AxiomC{$\mathcal{D}''$}
\noLine
\UnaryInfC{$\vdash U:B$}
\RightLabel{$\with$R$1$}
\TrinaryInfC{$x:B \vdash \mathtt{copy}^{U}  x \mathtt{\ as\ } x_1, x_2 \mathtt{\ in\ }\langle M_1, M_2 \rangle :B_1 \with B_2$}
\RightLabel{$cut$}
\BinaryInfC{$ \vdash \mathtt{copy}^{U} N
 \mathtt{\ as\ } x_1, x_2 \mathtt{\ in \ }\langle M_1, M_2\rangle :B_1\with B_2$}
\DisplayProof
\end{equation}
Since $B$ is a $\forall$-lazy type, $\mathcal{D}'$ is a  $\forall$-lazy derivation. By Lemma~\ref{lem: existence of the safe exponential cut}.\ref{enum: safe} and maximality of $h(R)$, $\mathcal{D}'$ has no critical cut, hence  $\mathcal{D}'$ is cut-free  by Lemma~\ref{lem: existence of the safe exponential cut}.\ref{enum: copy-first}. Therefore,  $R$ is a ready cut.  

As for point~\ref{enum: eliminating ready cut 2}, let  $\mathcal{D}^*$ be the derivation obtained by eliminating a ready cut like~\eqref{eqn: above the cut} in $\mathcal{D}$ (see Figure~\ref{fig: cut elimination LAM}).  By Remark~\ref{rem: eta long nf}, $\mathcal{D}''$ is $\eta$-expanded, hence cut-free by definition. Since both $\mathcal{D}'$ and $\mathcal{D}''$ are $\forall$-lazy and cut-free, they  have no instances of $\with$R$1$ by Proposition~\ref{prop: lazyness propagation}.\ref{point 2: laziness propagation}, so that $\with(\mathcal{D^*})= \with(\mathcal{D})-1$.  Moreover,  $\vert \mathcal{D}' \vert \leq \vert \mathcal{D}''\vert$ by Proposition~\ref{prop: properties of eta logn nf}.\ref{point: properties eta long nf 2}. We have: 
$2 \cdot \vert \mathcal{D}'\vert +\vert \mathcal{D}_1\vert+\vert \mathcal{D}_2\vert +3 + 2 \cdot \with(\mathcal{D^*}) < \vert \mathcal{D}\vert+2 \cdot \with(\mathcal{D^*})+ 2=\vert \mathcal{D}\vert+2 \cdot( \with(\mathcal{D}^*)+1)$. Therefore,  $\vert \mathcal{D}^* \vert +  2\cdot \with(\mathcal{D^*}) <\vert \mathcal{D}\vert+2 \cdot \with(\mathcal{D})$.
\end{proof}

\begin{thm}[Cubic $\forall$-lazy cut-elimination]
\label{thm: cut elimination for LAM}
Let $\mathcal{D}\triangleleft \Gamma \vdash M: A$ be a $\forall$-lazy derivation. 
Then, the $\forall$-lazy cut-elimination reduces $\mathcal{D}$ to a cut-free $\mathcal{D}^\dagger\triangleleft \Gamma \vdash M^\dagger:A$ in $\mathcal{O}(\vert \mathcal{D} \vert^3)$ steps.  
\end{thm}
\begin{proof}
Let us define a \textit{$\forall$-lazy cut-elimination strategy} divided into \textit{rounds}.   At each round:
\begin{enumerate}[\{1\}]
	\item \label{enum: round comm}  
	 we eliminate all the commuting instances of $ cut $;
	\item \label{enum: round sym} 
	if a  symmetric instance of $cut$ exists, we eliminate it; otherwise, all instances of $cut$ are either critical or copy-first, and we eliminate a \textit{ready} one, if any.
\end{enumerate}
We now show that the above $\forall$-lazy cut-elimination strategy terminates with a cut-free derivation.  
We proceed by induction on the lexicographical  order of the pairs 
$\langle \vert \mathcal{D}\vert+ 2\cdot \with(\mathcal{D}) , H(\mathcal{D}) \rangle$, where  $H(\mathcal{D})$ is the sum of the heights $h(\mathcal{D}')$ 
of all  subderivations $\mathcal{D}'$ of $\mathcal{D}$ whose conclusion is an instance of $ cut $. During~\ref{enum: round comm},  every commuting $\forall$-lazy cut-elimination step   moves an instance of $ cut $ upward, strictly decreasing $H(\mathcal{D})$ and leaving $ \vert \mathcal{D}\vert+2\cdot \with(\mathcal{D}) $ unaltered. During~\ref{enum: round sym}, every symmetric $\forall$-lazy cut-elimination step shrinks  $ \vert \mathcal{D}\vert $. If  only critical and copy-first cuts are in $\mathcal{D}$ then, by Lemma~\ref{lem: existence of the safe exponential cut}.\ref{enum: copy-first}, either $\mathcal{D}$ has critical cuts or it is cut-free. In the former case, by Lemma~\ref{lem: eliminating a lazy cut}.\ref{enum: eliminating ready cut 1} a ready cut exists. By  Lemma~\ref{lem: eliminating a lazy cut}.\ref{enum: eliminating ready cut 2}, if we  apply the corresponding $\forall$-lazy cut-elimination step $\mathcal{D}\red \mathcal{D}^*$, we have $ \vert \mathcal{D}^*\vert+ 2 \cdot \with(\mathcal{D}^*) < \vert \mathcal{D}\vert+ 2 \cdot \with(\mathcal{D}) $.

We now exhibit a bound on the number of 
$\forall$-lazy cut-elimination steps from $ \mathcal{D} $ to $ \mathcal{D}^\dagger $.
Generally speaking, we can represent a $\forall$-lazy cut-elimination strategy as:
\allowdisplaybreaks
\begin{align}
\label{align:cut-elimination bound diagram}
\mathcal{D}
= \mathcal{D}_0
\underset{cc_0}{\rightsquigarrow^*}
\mathcal{D}'_0
\underset{src_0}{\rightsquigarrow}
\mathcal{D}_1
\,\cdots \underset{cc_i}{\rightsquigarrow^*}
\mathcal{D}'_{i}
\underset{src_i}{\rightsquigarrow}
\mathcal{D}_{i+1}
\underset{cc_{i+1}}{\rightsquigarrow^*}\cdots\,
\mathcal{D}'_{n-1}
\underset{src_{n-1}}{\rightsquigarrow}
\mathcal{D}_n
\underset{cc_n}{\rightsquigarrow^*}
\mathcal{D}'_{n}
= \mathcal{D}^*
\end{align}
where,  for $ 0\leq i\leq n $ and $ 0\leq j\leq  n-1 $,  $ cc_i $ denotes a sequence of $\forall$-cut elimination steps applied to  commuting cuts, while $src_j$ denotes a  $\forall$-cut elimination step applied to either a symmetric  or a ready cut. A bound on the length of the sequence $ cc_i $ is $ \vert\mathcal{D}_{i}\vert^2 $ because every instance of rule in $ \mathcal{D}_{i} $ can, in principle, be commuted with every other. Moreover, $ \vert\mathcal{D}_{j+1}\vert  + 2 \cdot\with(\mathcal{D}_{j+1})< \vert\mathcal{D}'_{j}\vert + 2 \cdot\with(\mathcal{D}'_j)$,
for every $ 0\leq j\leq n-1 $. Finally, since  $ \vert\mathcal{D}_{i}\vert = \vert\mathcal{D}'_{i}\vert $ for all $0 \leq i \leq n$, we have $ n\leq \vert\mathcal{D}\vert + 2\cdot \with (\mathcal{D}) \leq 3 \cdot \vert \mathcal{D}\vert$. Therefore,  the total number of $\forall$-lazy cut-elimination steps in 
\eqref{align:cut-elimination bound diagram} is 
$ O(\vert\mathcal{D}\vert\cdot\vert\mathcal{D}\vert^2) $.
\end{proof}


Recalling that the normal form of a typable term in $\mathsf{LAM}$ exists by Theorem~\ref{thm:lazy normalization for LAML} and is unique by Remark~\ref{rem: uniqueness nf}, we have the following straightforward corollary.

\begin{cor}\label{cor: confluence and progress} Let $\mathcal{D}\triangleleft \Gamma  \vdash M:A$ be  a $\forall$-lazy derivation, and let  $M^\dagger$ be  the normal form of $M$. Then:
	\begin{itemize}
		\item 	 $M^\dagger$  is free from instances of $\mathtt{copy}$ and $\pi_i$;
		\item  if $\Gamma= \emptyset$ then $M^\dagger\in \mathcal{V}$. 
	\end{itemize}
\end{cor}
\begin{proof}
	It suffices to observe that,   if a $\forall$-lazy derivation $\mathcal{D}\triangleleft \Gamma \vdash M: A$  rewrites to $\mathcal{D}'\triangleleft \Gamma \vdash M': A$ by a $\forall$-lazy cut-elimination step, then $M\rightarrow^* M'$ (we consider terms up to $\alpha$-equivalence). By Theorem~\ref{thm: cut elimination for LAM} a cut-free $\mathcal{D}^\dagger$ exists such that  $\mathcal{D}^\dagger\triangleleft \Gamma \vdash M^\dagger: A$, for some $M^\dagger$.  We conclude by Proposition~\ref{prop: lazyness propagation}.\ref{point 2: laziness propagation}.
\end{proof}

As we already observed in Section~\ref{subsec: linear additives},  a   $\mathtt{ copy}$ construct  behaves  quite like a  suspended substitution. So, a normal form with shape $\mathtt{copy}^{V} \, N  \mathtt{ \ as \ }x_1, x_2 \mathtt{ \ in\  }\langle P, Q \rangle$  represents a substitution that cannot be performed. Corollary~\ref{cor: confluence and progress} states that, whenever a term has $\forall$-lazy type, its normal form is always free from these \enquote{unevaluated} expressions. This result is analogous to Girard's linear normalization by lazy evaluation for terms having lazy type in $\mathsf{IMALL}_2$ (see Section~\ref{subsec:blowuplazy}). However, the above corollary allows us to gain something more.  On the one hand, indeed,  $\mathsf{LAM}$'s typable terms enjoy linear \textit{strong} normalization (Theorem~\ref{thm:lazy normalization for LAML}).  Therefore, as opposed to $\mathsf{IMALL}_2$,  the system $\mathsf{LAM}$ does not require specific evaluation strategies to avoid exponential reductions. On the other hand, as already remarked,   the  $\forall$-lazy types are more expressive than the  lazy ones (see Figure~\ref{fig: foralllazy}).

\section{Comparing $\mathsf{LAM}$ and $\mathsf{IMLL}_2$}

Following~\cite{curzi2019type}, we exploit the mechanisms of  linear erasure and duplication studied by Mairson and Terui~\cite{mairsonlinear, mairson2003computational} to define a sound translation of $\mathsf{LAM}$ into $\mathsf{IMLL}_2$ (Theorem~\ref{thm: translations for LAM}). A  fundamental result of this section is Theorem~\ref{thm: exponential ompression for LAM}, stating that derivations of $\mathsf{LAM}$ may exponentially compress linear $\lambda$-terms of $\mathsf{IMLL}_2$. On the one hand, these results witness that the former system is not algorithmically more expressive than the latter. On the other hand, in a way similar to~\cite{curzi2019type},  they show that $\mathsf{LAM}$ is able to compactly represent Mairson and Terui's linear erasure and duplication.  

\subsection{Linear Erasure and Duplication in $\mathsf{IMLL}_2$}

Mairson has shown in~\cite{mairsonlinear} that $\mathsf{IMLL}$ is expressive enough to encode boolean circuits. This result was later reformulated by Mairson and Terui  in $\mathsf{IMLL}_2$  to prove results about the complexity of cut-elimination 
\cite{mairson2003computational}, where the advantage of working with $\mathsf{IMLL}_2$   is to assign uniform types to structurally related linear $\lambda$-terms.  In the latter encoding,  the boolean values \enquote{true} and \enquote{false} are  represented by $\mathtt{tt}\triangleq	\lambda x. \lambda y.  x\otimes y  $   and  $\mathtt{ff}\triangleq \lambda x. \lambda y.  y\otimes x $ respectively, with type $\mathbf{B}\triangleq\forall \alpha. \alpha \multimap \alpha \multimap \alpha \otimes \alpha$. The key step of the encoding is the existence of an  \enquote{eraser} $\mathtt{E}_{\mathbf{B}}$ and a \enquote{duplicator} $\mathtt{D}_{\mathbf{B}}$ for the Boolean data type  $\mathbf{B}$:
 \allowdisplaybreaks
\begin{align}
\label{eqn: erasure booleans}
\mathtt{E}_{\mathbf{B}}& \triangleq
\lambda z. \mathtt{let\ }z\mathbf I\mathbf I \mathtt{\ be \ } x\otimes y \mathtt{ \ in \ }(\mathtt{let \ }y \mathtt{ \ be \ } \mathbf I \mathtt{ \ in \ }x) :\mathbf{B}\multimap \mathbf{1}
\\
\label{eqn: duplication booleans}
\mathtt{D}_{\mathbf{B}}& \triangleq
\lambda z. \mathtt{proj}_1 (z (\mathtt{tt}\otimes \mathtt{tt}) (\mathtt{ff}\otimes \mathtt{ff})) : \mathbf{B}\multimap \mathbf{B}\otimes \mathbf{B}
\\
\label{eqn: boolean projection}
\mathtt{proj}_1 &\triangleq
\lambda z. \mathtt{let\ }z \mathtt{\ be \ }x\otimes y \mathtt{\ in\ }(\mathtt{let \ }\mathtt{E}_{\mathbf{B}} \, y \mathtt{\ be \ }\mathbf I \mathtt{\ in \ }x) :
(\mathbf{B}\otimes \mathbf{B}) \multimap\mathbf{B}
\end{align}
where $\mathtt{proj}_1$ is the linear $\lambda$-term projecting the first element of a pair. For $M\in \{ \mathtt{ tt} , \mathtt{ff}\}$, we have $ \mathtt{E}_{\mathbf{B}} \,M \rightarrow_{\beta}^* \mathbf{I}$ and $\mathtt{D}_{\mathbf{B}}\,M\rightarrow^*_\beta  M\otimes  M$. In other words,  linear erasure involves a stepwise \enquote{data consumption} process, while linear duplication  works  \enquote{by selection and erasure}: it contains  \textit{both} possible outcomes of duplication $ \mathtt{tt}\otimes \mathtt{tt} $ and $ \mathtt{ff}\otimes \mathtt{ff} $, and it selects the desired pair   by linearly erasing the other one.

In~\cite{mairson2003computational}, Mairson and Terui  generalize the above mechanism of linear erasure and duplication to the class of closed $\Pi_1$ types:

\begin{defn}[$\Pi_1$ types~\cite{mairson2003computational}]
\label{defn:Pi and Sigma types} 
A type of $\mathsf{IMLL}_2$ is a $\Pi_1$ type if it contains no negative occurrences of  $\forall$. 
\end{defn}

Closed  $\Pi_1$ types represent \emph{finite} data types, because they admit only finitely many   closed and normal inhabitants. An example is $\mathbf{B}$, representing the Boolean data type.

The fundamental result about closed $\Pi_1$ types is the following:

\begin{thm}[Erasure and duplication~\cite{mairson2003computational}]\label{thm: duplication and erasure} 
{\ }
\begin{enumerate}
\item \label{point: erasure} For any closed $\Pi_1$ type $A$ there is a  linear $\lambda$-term $\mathtt{E}_{A}$  of type $A\multimap \mathbf{1}$ such that, for all closed and normal inhabitant $M$ of $A$, $\mathtt{E}_{A} \, M \rightarrow_{\beta  }^* \mathbf I$.
\item \label{point: duplication} For any closed  and \emph{inhabited} $\Pi_1$ type $A$ there is a  linear $\lambda$-term $\mathtt{D}_{A}$  of type $A\multimap A \otimes A$ such that, for all closed and normal inhabitant $M$ of $A$,  $\mathtt{D}_{A} \,M \rightarrow_{\beta \eta}^* 
 M\otimes M $.
\end{enumerate}
\end{thm}
We call $\mathtt{E}_A$ \textit{eraser} and $\mathtt{D}_A$ \textit{duplicator} of $A$. 
 Intuitively, by taking as input a closed and normal inhabitant  $M$ of closed $\Pi_1$ type $A$, $\mathtt{D}_A$ implements the following three main operations:
\begin{enumerate}[(1)]
	\item \label{enumerate: duplication step 1}
	\enquote{expand} $ M $ to an $ \eta $-long normal form of $A$, let us say $M'$;
	
	\item \label{enumerate: duplication step 2}
    compile $ M' $ to a linear $ \lambda $-term
    $ \lceil M' \rceil $ which encodes $ M' $ as 
    a boolean tuple;
	
	\item \label{enumerate: duplication step 3}
    copy and decode $ \lceil M' \rceil $ obtaining 
    $  M'\otimes M'  $, which $\eta$-reduces to $  M\otimes M  $.  
\end{enumerate}
Point~\ref{enumerate: duplication step 3} implements Mairson and Terui's \enquote{duplication by selection and erasure} discussed for the type $\mathbf{B}$, and requires  a nested  series of  \texttt{if}-\texttt{then}-\texttt{else}  playing the role of a look-up table that stores all  pairs of closed and normal inhabitants of $A$  (which are always finite, as already observed). Each pair represents a possible outcome of  duplication. Given a boolean tuple $ \lceil M' \rceil $  as input, the nested \texttt{if}-\texttt{then}-\texttt{else} select the corresponding pair 
    $  M'\otimes  M'  $, erasing all the remaining \enquote{candidates}. The inhabitation condition for $A$ stated in Theorem~\ref{thm: duplication and erasure}.\ref{point: duplication} assures the existence of a default pair $  N\otimes N  $,  a sort of \enquote{exception} that we ``throw'' if the boolean tuple in input  encodes no closed normal inhabitant of  $A$.
    
Point~\ref{point: duplication} of Theorem~\ref{thm: duplication and erasure} was only   sketched in~\cite{mairson2003computational}. A detailed proof of the construction of   $\mathtt{D}_A$ is in~\cite{curzi2019type}, which also estimates the complexity of duplicators and erasers:

\begin{prop}[Size  of $\mathtt{E}_A$ and  $\mathtt{D}_{A}$~\cite{curzi2019type}] 
\label{prop: size of duplicator}  If  $A$ is a closed $\Pi_1$ type, then $\vert \mathtt{E}_{A} \vert \in \mathcal{O}(\vert A \vert)$. Moreover, if $A$ is inhabited, then $\vert \mathtt{D}_{A} \vert \in \mathcal{O}(  2^{\vert A \vert^2})$.
\end{prop}

\subsection{A  Translation of $\mathsf{LAM}$ Into $\mathsf{IMLL}_2$ and Exponential Compression}
Following~\cite{curzi2019type}, we  define  a translation $(\_)^\bullet$ from derivations of $\mathsf{LAM}$ into linear $\lambda$-terms with type in $\mathsf{IMLL}_2$. It   maps  closed $\forall$-lazy types into closed $\Pi_1$ types, and  instances of the  inference rules $\with$R$1$ and $\with$L$i$ into, respectively,  duplicators and erasers of closed $\Pi_1$ types. We prove that the translation is sound and   the linear $\lambda$-term $\mathcal{D}^\bullet$ associated with a derivation  $\mathcal{D}$ of $\mathsf{LAM}$ has  size that can be  exponential with respect to the size of $\mathcal{D}$.

\begin{defn} [From $ \mathsf{LAM} $ to $ \mathsf{IMLL}_2 $]
\label{defn: translation LAM  into IMLL2} 
We define a map $(\_)^\bullet$  translating  a derivation 
$\mathcal{D} \triangleleft\Gamma \vdash  M: A $ of $\mathsf{LAM}$ into a linear $\lambda$-term $\mathcal{D}^\bullet$ such that  $\Gamma^\bullet \vdash \mathcal{D}^\bullet :A^\bullet$ is derivable in $\mathsf{IMLL}_2$.
\begin{enumerate}
\item The map $(\_)^\bullet$  is defined on types of $\Theta_\with$ by induction on their structure:
\allowdisplaybreaks
\begin{align*}
\alpha^\bullet &\triangleq \alpha\  &  (A \with B)^\bullet &\triangleq A^\bullet \otimes B^\bullet\\
(A \multimap B)^\bullet &\triangleq A^\bullet \multimap B^\bullet & (\forall \alpha . A)^\bullet &\triangleq \forall \alpha. A^\bullet  
\enspace .
\end{align*}
Notice that $(A \langle B/ \alpha \rangle)^\bullet = A^\bullet \langle B^\bullet/ \alpha \rangle$.  If  $\Gamma= x_1: A_1, \ldots, x_n: A_n$,  we set $\Gamma^\bullet\triangleq  x_1: A_1^\bullet, \ldots, x_n: A_n^\bullet$.
\item \label{point: translation on derivations} 
The map $(\_)^\bullet$ is defined on derivations $\mathcal{D}\triangleleft \Gamma \vdash M:A$ of $\mathsf{LAM}$ by induction on the last rule:
\begin{enumerate}[(a)]
\item if   $\mathcal{D}$ is $ax$ with conclusion $x: A \vdash x:A$, then $\mathcal{D}^\bullet\triangleq x$;
\item if  $\mathcal{D}$  has last rule $cut$ with premises $\mathcal{D}_1\triangleleft \Delta \vdash N:B$ and $\mathcal{D}_2\triangleleft \Sigma, x:B \vdash P: A$, where  $M=P[N/x]$,   then  $\mathcal{D}^\bullet\triangleq \mathcal{D}_2^\bullet [\mathcal{D}_1^\bullet/x]$;
\item  if $\mathcal{D}$ has last rule $\multimap$R with premise $\mathcal{D}_1 \triangleleft \Gamma,x: B \vdash N:C$ and $M= \lambda x.N$,   then $\mathcal{D}^\bullet\triangleq \lambda x. \mathcal{D}_1^\bullet$; 
\item if $\mathcal{D}$ has last rule $\multimap$L with premises $\mathcal{D}_1\triangleleft \Delta \vdash N:B$ and $\mathcal{D}_2\triangleleft \Sigma, x:C \vdash P: A$, where $\Gamma= \Delta, \Sigma, y: B \multimap C$, then $\mathcal{D}^\bullet \triangleq \mathcal{D}_2^\bullet[y \mathcal{D}^\bullet _1/x]$;
\item  if $\mathcal{D}$ has last rule $\with$R$0$ with premises $\mathcal{D}_1\triangleleft  \vdash N_1:B_1$ and $\mathcal{D}_2\triangleleft  \vdash N_2: B_2$ then $\mathcal{D}^\bullet \triangleq  \mathcal{D}_1^\bullet\otimes \mathcal{D}_2^\bullet $;
\item \label{point: translation eraser} if  $\mathcal{D}$ has last rule  $\with$L$i$ with premise $\mathcal{D}_1\triangleleft \Delta, x_i:B_{i}  \vdash N:A$, where $\Gamma= \Delta, x: B_1 \with B_2$,   then $\mathcal{D}^\bullet \triangleq\mathtt{let \ }x \mathtt{ \ be\  }x_1\otimes  x_2 \mathtt{ \ in \ } (\mathtt{let \ } \mathtt{E}_{B^\bullet_{3-i}} x_{3-i} \mathtt{ \ be \ }\mathbf{I} \mathtt{ \ in \ }\mathcal{D}^\bullet_1)$, where  $\mathtt{E}_{B^\bullet_{3-i}}$ is the eraser of  $B^\bullet_{3-i}$;
\item \label{point: translation duplicator} if  $\mathcal{D}$ ends with $\with$R$1$ with premises $\mathcal{D}_1\triangleleft x_1: B \vdash N_1:B_1$,  $\mathcal{D}_2\triangleleft x_2:B \vdash N_2: B_2$, and $\mathcal{D}_3 \triangleleft \vdash V^\bullet :B^\bullet$ then  $\mathcal{D}^\bullet\triangleq \mathtt{ \ let \ }\mathtt{D}_{B^\bullet} x \mathtt{\ be \ }x_1\otimes  x_2  \mathtt{ \ in \ }  \mathcal{D}_1^\bullet \otimes \mathcal{D}_2^\bullet$, where $\mathtt{D}_{B^\bullet}$ is the duplicator  of $B^\bullet$;
\item   if  $\mathcal{D}$ has last rule $\forall$R with premise $\mathcal{D}_1 \triangleleft \Gamma \vdash M:B\langle \gamma / \alpha \rangle$  then $\mathcal{D}^\bullet \triangleq \mathcal{D}_1^\bullet$;
\item   if $\mathcal{D}$ has last rule $\forall$L with premise $\mathcal{D}_1 \triangleleft \Delta, x: B \langle C/ \alpha \rangle \vdash M:A$, where $\Gamma= \Delta, x:\forall \alpha. B$,  then $\mathcal{D}^\bullet\triangleq \mathcal{D}^\bullet _1$.
\end{enumerate}
\end{enumerate}   
\end{defn}
 
\begin{rem}\label{rem: translation well-defined}  Points~\ref{point: translation on derivations}\ref{point: translation eraser}-\ref{point: translation duplicator} are well-defined.  Indeed, since   $B,  B_1, B_2$ in both points are closed and $\forall$-lazy,  the types $B^\bullet,  B_1^\bullet, B_2^\bullet$  are  closed   $\Pi_1$, so that Theorem~\ref{thm: duplication and erasure}.\ref{point: erasure} assures the existence of  $\mathtt{E}_{B_{3-i}^\bullet}$. Moreover, the closed $\Pi_1$ type $B^\bullet$ in  point~\ref{point: translation on derivations}\ref{point: translation duplicator}  is  inhabited by the closed linear $\lambda$-term $\mathcal{D}_3^\bullet$. The latter is also normal, since by Remark~\ref{rem: eta long nf}    $\mathcal{D}_3$  is $\eta$-expanded (hence cut-free). Therefore, Theorem~\ref{thm: duplication and erasure}.\ref{point: duplication} assures that  $\mathtt{D}_{B^\bullet}$ exists.
\end{rem}

We now show that every  $\forall$-lazy cut-elimination step applied to a derivation $\mathcal{D}\triangleleft \Gamma \vdash M: A$  of  $\mathsf{LAM}$ can be simulated  by a sequence of $\beta \eta$-reduction steps  applied to $\mathcal{D}^\bullet$.  

\begin{thm}[Soundness of $(\_)^\bullet$]
\label{thm: translations for LAM} Let $\mathcal{D}$ be a derivation of $\mathsf{LAM}$. If $\mathcal{D}\red \mathcal{D}'$ then $\mathcal{D}^\bullet\rightarrow^*_{\beta \eta}\mathcal{D}'^\bullet$.
\end{thm}
\begin{proof} W.l.o.g.~it suffices to consider the case where the last rule of $\mathcal{D}$ is the instance of $cut$   the $\forall$-lazy cut-elimination rule $\mathcal{D}\red \mathcal{D}'$  is applied to. 
The only interesting cases are the $\forall$-lazy cut-elimination rules in Figure~\ref{fig: cut elimination LAM}.
So, suppose that $\mathcal{D}$ ends with a cut $(\with\text{R}0, \with\text{L}i)$, where the premises of $\with\text{R}0$ are   $\mathcal{D}_1\triangleleft \vdash N_1: A_1$ and $\mathcal{D}_2\triangleleft \vdash N_2: A_2$, and the premise of $\with$L$i$ is $\mathcal{D}_3 \triangleleft \Gamma, x_i: A_i \vdash M:B$. Since $\mathcal{D}_{3-i}^\bullet$ is a closed  linear $\lambda$-term of closed $\Pi_1$ type $A^\bullet_{3-i}$,  by applying Theorem~\ref{thm: duplication and erasure}.\ref{point: erasure}   and the reduction rules in~\eqref{eqn: tensor unit reduction rules} we have:
\allowdisplaybreaks
\begin{align*}
 \mathcal{D}^\bullet&= \mathtt{ \ let \ }  \mathcal{D}_1^\bullet\otimes  \mathcal{D}_2^\bullet  \mathtt{ \ be \ }x_1\otimes  x_2 \mathtt{ \ in \ }(\mathtt{ let \ } \mathtt{E}_{A^\bullet_{3-i}} x_{3-i} \mathtt{ \ be \ } \mathbf{I} \mathtt{ \ in\ } \mathcal{D}_3^\bullet)\\
&\rightarrow_\beta \mathtt{ let \ } \mathtt{E}_{A_{3-i}} \mathcal{D}^\bullet_{3-i} \mathtt{ \ be \ } \mathbf{I} \mathtt{ \ in\ } \mathcal{D}_3^\bullet[\mathcal{D}_i^\bullet/x_i]\\ 
&\rightarrow^*_{\beta } \mathtt{ let \ }\mathbf{I} \mathtt{ \ be \ } \mathbf{I} \mathtt{ \ in\ } \mathcal{D}_3^\bullet[\mathcal{D}_i^\bullet/x_i]\\ 
&\rightarrow_{\beta }\mathcal{D}_3^\bullet[\mathcal{D}_i^\bullet/x_i]= \mathcal{D}'^\bullet.
\end{align*}
Finally, suppose that 
$\mathcal{D}$ ends with a \textit{ready} cut $(X, \with\text{R}1)$, for some $X$,  where the left premises of the $cut$  is $\mathcal{D}_1\triangleleft \vdash V:A$ and the premises of  $\with\text{R}1$ are   $\mathcal{D}_2\triangleleft x_1:A \vdash N_1: B_1$,  $\mathcal{D}_3\triangleleft x_2:A \vdash N_2: B_2$ and $\mathcal{D}_4\triangleleft \vdash U:A$.  Since the cut is ready,    $\mathcal{D}_1$ must be cut-free, and hence $\mathcal{D}_1^\bullet$ is a closed and normal linear $\lambda$-term of closed $\Pi_1$ type $A^\bullet$. Therefore, by applying    Theorem~\ref{thm: duplication and erasure}.\ref{point: duplication} and the reduction rules in~\eqref{eqn: tensor unit reduction rules},  we have: 
\allowdisplaybreaks
\begin{align*}
\mathcal{D}^\bullet&= \mathtt{ \ let \ }\mathtt{D}_{A^\bullet} \mathcal{D}_1^\bullet    \mathtt{\ be \ }x_1\otimes  x_2  \mathtt{ \ in \ }  \mathcal{D}_2^\bullet\otimes \mathcal{D}_3^\bullet\\
&\rightarrow^*_{\beta \eta} \mathtt{ \ let \ }  \mathcal{D}_1^\bullet\otimes \mathcal{D}_1^\bullet   \mathtt{\ be \ }x_1\otimes  x_2  \mathtt{ \ in \ }  \mathcal{D}_2^\bullet\otimes  \mathcal{D}_3^\bullet\\
&\rightarrow_\beta \mathcal{D}_2^\bullet[\mathcal{D}_1^\bullet/x_1]\otimes  \mathcal{D}_3^\bullet [\mathcal{D}_1^\bullet/x_2]=\mathcal{D}'^\bullet 
\end{align*}
this concludes the proof.
\end{proof}

The above result stresses a fundamental advantage of  $\mathsf{LAM}$ over $\mathsf{IMLL}_2$: as shown in both~\cite{mairson2003computational} and~\cite{curzi2019type},   the latter type system  requires an extremely complex linear $\lambda$-term to encode linear duplication  (see Theorem~\ref{prop: size of duplicator}), while  the former   can  compactly represent it by means of typed terms with shape:
\begin{equation}
\lambda x.  \mathtt{copy}^{V}  x \mathtt{\ as \ }x_1, x_2 \mathtt{ \ in\ } \langle x_1,x_2 \rangle: A \multimap  A \with A
\end{equation}
Moreover,  linear erasure is expressed by the following simple typed term:
\begin{equation}
 \lambda x.  \pi_2  (\langle  x, \mathbf{I}\rangle)  : A \multimap \mathbf{1}
\end{equation}

This crucial aspect of $\mathsf{LAM}$ can be made apparent by estimating the impact of the translation $(\_)^\bullet$ on the size of derivations, and hence the cost of   \enquote{unpacking}  the inference rules $\with$R$1$ and $\with$L$i$.
\begin{thm}[Exponential compression for $\mathsf{LAM}$] \label{thm: exponential ompression for LAM}Let $\mathcal{D}$ be a derivation in $\mathsf{LAM}$. Then, $\vert \mathcal{D}^\bullet \vert = \mathcal{O}(2^{\vert \mathcal{D} \vert^k})$, for some $k \geq 1$.
\end{thm}
\begin{proof} By structural induction on $\mathcal{D}$. The only  interesting case is when $\mathcal{D}$ ends with $\with$R$1$ with premises  $\mathcal{D}_1\triangleleft x_1:A \vdash N_1: B_1$,  $\mathcal{D}_2\triangleleft x_2:A \vdash N_2: B_2$ and $\mathcal{D}_3\triangleleft \vdash U:A$. By Definition~\ref{defn: translation LAM  into IMLL2},  $\mathcal{D}^\bullet = \mathtt{ \ let \ }\mathtt{D}_{A^\bullet} x    \mathtt{\ be \ }x_1\otimes  x_2  \mathtt{ \ in \ }  \mathcal{D}_1^\bullet\otimes  \mathcal{D}_2^\bullet $. 
By Remark~\ref{rem: eta long nf}, $\mathcal{D}_3$ is $\eta$-expanded, so that $\vert A \vert \leq 2 \cdot \vert \mathcal{D}_3 \vert $ by Proposition~\ref{prop: properties of eta logn nf}.\ref{point: properties eta long nf 1}. Hence,  $\vert \mathtt{D}_{A} \vert \in \mathcal{O}(  2^{(2\cdot \s{\mathcal{D}_3})^2})$ by Proposition~\ref{prop: size of duplicator}. We apply the induction hypothesis on $\mathcal{D}_1$ and $\mathcal{D}_2$ and conclude.
\end{proof}

\section{Conclusions}
We  introduce  $\mathsf{LAM}$,  a type assignment system endowed with  a weaker version of the Linear Logic additive rules $\with$R and $\with$L, called \emph{linear additive rules}.  We prove both  linear strong  normalization and a restricted  cut-elimination theorem. Also, we  present a sound  translation of  $\mathsf{LAM}$ into   $\mathsf{IMLL}_2$,  and we study its complexity.

A future direction is to find linear additive rules  based on the additive connective $\oplus$,  and to prove results similar to Theorem~\ref{thm:lazy normalization for LAML} and Theorem~\ref{thm: cut elimination for LAM}. This goal turns out to be harder, due to the \enquote{classical flavor} of the inference rule $\oplus$L, displayed below:
\begin{prooftree}
\AxiomC{$\Gamma, x:A \vdash M:C$}
\AxiomC{$\Gamma, y:B \vdash  N:C$}
\BinaryInfC{$\Gamma, z:A \oplus B \vdash \mathtt{case}\ z \mathtt{\ of \ }[\mathtt{inj}_1(z)\to M \ \vert \ \mathtt{inj}_2(z)\to N]:C$}
\end{prooftree}
 Let us discuss this point. The linear additive rule $\with$R$1$ prevents  exponential normalization by carefully controlling context-sharing, which involves hidden contractions and is responsible for unrestricted duplication. Finding a linear additive rule corresponding to $\oplus$L means controlling the sharing of types in the right-hand side of the turnstile. But this sharing hides a co-contraction, i.e.~$C \otimes C \multimap C$, which requires  fairly different techniques to be tamed. 

Interesting applications of linear additives are in the field of ICC, and indeed they motivate the tools developed in the present paper. As already discussed in the Introduction, variants of the additive rules expressing non-determinism explicitly  have been used to capture  $\mathsf{NP}$~\cite{matsuoka2004nondeterministic, gaboardi2008soft, maurel2003nondeterministic}.  To the best of our knowledge, all these  characterizations of $\mathsf{NP}$   crucially depend on the choice of a special evaluation strategy  able to avoid  the exponential blow up  described in  Section~\ref{subsec:blowuplazy}.  Linear additives can refine~\cite{matsuoka2004nondeterministic, gaboardi2008soft, maurel2003nondeterministic}, because they do not affect the complexity of normalization, and so they allow for natural cost models that can be implemented with a negligible overhead.  A possible future work could be then to extend Soft Type Assignment ($\mathsf{STA}$), a type system capturing $\mathsf{PTIME}$~\cite{gaboardi2009light},  with a non-deterministic variant of linear additives, and to show that   \textit{Strong} Non-deterministic Polytime Soundness holds for the resulting system. This would allow us to characterize $\mathsf{NP}$  in a \enquote{wider} sense, i.e.~independently of the reduction strategy considered.

In a probabilistic setting, similar goals have  already been achieved. In~\cite{curzi2020probabilistic} Curzi and Roversi studied the type system $\mathsf{PSTA}$, an extension of $\mathsf{STA}$ with a non-deterministic variant of the linear additive rules obtained by  replacing  $\with$L$i$  with the following:
\begin{equation*}
	\AxiomC{$\Gamma, x:A \vdash M:C$}
	\UnaryInfC{$\Gamma, y:A \with A \vdash M[\pi(y)/x]:C$}
	\DisplayProof
\end{equation*}
and by considering the  non-deterministic reduction rule $M \leftarrow \pi (\langle M,N \rangle)\rightarrow N$ in  place of $\pi_i(\langle M_1,M_2 \rangle)\rightarrow M_i$. It is shown that, when  $\mathsf{PSTA}$ is  endowed with a probabilistic big-step reduction relation, it is able to capture the probabilistic polytime functions and problems independently of the reduction strategy.

\section*{Acknowledgments}
I would like to thank L. Roversi for the precious discussions about the topic, and the anonymous reviewers for useful comments and suggestions. This work was supported by a UKRI Future Leaders Fellowship, `Structure vs Invariants in Proofs', project reference MR/S035540/1.

\nocite{*}
\bibliographystyle{eptcs}
\bibliography{generic}
\end{document}